\documentclass{article}
\usepackage{geometry}
\usepackage[nospace,noadjust]{cite}
\usepackage{amsmath,amssymb,amsfonts,times}
\usepackage{graphicx}
\usepackage{textcomp}
\usepackage{xcolor}
\usepackage{subfigure}
\usepackage{epsfig}
\usepackage{multirow}
\usepackage{algorithm}
\usepackage{algorithmicx}
\usepackage{algpseudocode}
\usepackage{array}
\usepackage{epstopdf}
\usepackage{color}
\usepackage{diagbox}
\usepackage{bm}
\usepackage{bbm}
\usepackage{mathrsfs}
\usepackage{tcolorbox}
\usepackage{mathtools}
\usepackage{mdframed}
\usepackage{pdfsync}

\usepackage{url}
\usepackage[hidelinks]{hyperref}
\usepackage{amsmath,amsthm,amssymb,amsbsy}
\usepackage{paralist}
\usepackage{xcolor}
\usepackage{color}
\usepackage{comment}
\usepackage{multirow}
\usepackage[toc, page]{appendix}

\usepackage{fancyhdr}
\usepackage{cite}
\usepackage{cleveref}

\newtheorem{lemma}{Lemma}

\newtheorem{theorem}{Theorem}

\theoremstyle{remark}

\theoremstyle{problem}

\newcommand{\R}{\mathbb{R}}
\def \real    { \mathbb{R} }

\newcommand{\C}{\mathbb{C}}

\newcommand{\e}{\begin{equation}}
\newcommand{\ee}{\end{equation}}
\newcommand{\en}{\begin{equation*}}
\newcommand{\een}{\end{equation*}}
\newcommand{\eqn}{\begin{eqnarray}}
\newcommand{\eeqn}{\end{eqnarray}}
\newcommand{\bmat}{\begin{bmatrix}}
\newcommand{\emat}{\end{bmatrix}}
\renewcommand{\Re}[1]{\operatorname{Re}\left\{#1\right\}}
\renewcommand{\Im}[1]{\operatorname{Im}\left\{#1\right\}}

\DeclareMathAlphabet\mathbfcal{OMS}{cmsy}{b}{n}

\newcommand{\E}{\operatorname{\mathbb{E}}}

\newcommand{\vct}[1]{\boldsymbol{#1}}
\newcommand{\mtx}[1]{\boldsymbol{#1}}

\newcommand{\trace}{\operatorname{trace}}

\newcommand{\sign}{\operatorname{sign}}

\def \vec       {\operatorname*{vec}}
\newcommand{\wh}{\widehat}

\newcommand{\wt}{\widetilde}
\newcommand{\ol}{\overline}

\newcommand{\calB}{\mathcal{B}}

\newcommand{\calG}{\mathcal{G}}

\newcommand{\calN}{\mathcal{N}}

\newcommand{\calP}{\mathcal{P}}

\newcommand{\va}{\vct{a}}
\newcommand{\vb}{\vct{b}}

\newcommand{\vd}{\vct{d}}
\newcommand{\ve}{\vct{e}}
\newcommand{\vf}{\vct{f}}

\newcommand{\vw}{\vct{w}}
\newcommand{\vx}{\vct{x}}
\newcommand{\vy}{\vct{y}}
\newcommand{\vz}{\vct{z}}

\newcommand{\vtheta}{\vct{\theta}}

\newcommand{\vphi}{\vct{\phi}}

\newcommand{\vrho}{\vct{\rho}}

\newcommand{\mA}{\mtx{A}}
\newcommand{\mB}{\mtx{B}}

\newcommand{\mH}{\mtx{H}}

\newcommand{\mK}{\mtx{K}}

\newcommand{\mO}{\mtx{O}}
\newcommand{\mP}{\mtx{P}}
\newcommand{\mQ}{\mtx{Q}}

\newcommand{\mU}{\mtx{U}}
\newcommand{\mV}{\mtx{V}}
\newcommand{\mW}{\mtx{W}}
\newcommand{\mX}{\mtx{X}}
\newcommand{\mY}{\mtx{Y}}
\newcommand{\mZ}{\mtx{Z}}

\newcommand{\mTheta}{\mtx{\Theta}}

\newcommand{\mId}{{\bf I}}

\graphicspath{{./figs/}}

\newlength{\imgwidth}
\newboolean{twoColVersion}
\setboolean{twoColVersion}{false}
\newcommand{\twoCol}[2]{\ifthenelse{\boolean{twoColVersion}} {#1} {#2} }

\usepackage{ulem}

\usepackage{mathtools}

\title{\LARGE \bf Generalization of Transformer-Based Neural Quantum States via In-Context Learning}

\author{Zhen Qin, Qing Qu  and Alfred O. Hero III \thanks{ZQ (e-mail: zhenqin@umich.edu) is with the Michigan Institute for Computational Discovery and Engineering, Department of Electrical Engineering and Computer Science and Department of Statistics, University of Michigan, Ann Arbor, MI 48109 USA; QQ (email: qingqu@umich.edu) is with the Department of Electrical Engineering and Computer Science, University of Michigan, Ann Arbor, MI 48109 USA; AOH (email: hero@eecs.umich.edu) is with the Department of Electrical Engineering and Computer Science, University of Michigan, Ann Arbor, MI 48109 USA.}}

\begin{document}

\maketitle


\begin{abstract}

Neural quantum states based on modern deep learning architectures have emerged as powerful representations for quantum many-body systems. In particular, Transformer-based neural quantum states provide expressive models capable of capturing long-range correlations, and their empirical generalization performance has recently been demonstrated. However, a theoretical understanding of their generalization behavior remains largely unexplored.  In this paper, we develop a theoretical framework to analyze the generalization properties of Transformer-based neural quantum states under in-context learning. We establish a rigorous inference-time generalization error bound in terms of mean squared error (MSE), showing that the pointwise prediction error decreases inversely with both the number of in-context examples and the depth of the Transformer. We further show that the Transformer depth required to achieve this guarantee scales only linearly with the system size--namely, the number of particles in continuous systems or the number of qudits in discrete systems. Building on this result, we extend our analysis to full quantum states formulated as rank-one density operators, and derive MSE-based generalization bounds over both continuous and discrete domains under physical constraints. Finally, numerical simulations corroborate our theoretical analysis.
\end{abstract}

\section{Introduction}
\label{intro}

Quantum computing holds the promise of addressing computational tasks that are widely believed to be intractable for classical computers. Prominent examples include the simulation of quantum chemistry \cite{levine2009quantum}, the optimization of complex systems \cite{abel2022quantum}, and the preparation of ground states of many-body quantum systems \cite{poulin2009preparing}. A central obstacle in classical simulation arises from the exponential growth of the Hilbert space dimension with system size, which rapidly renders exact representations of quantum states computationally prohibitive. To mitigate this challenge, a variety of variational approaches have been developed, in which the quantum state is approximated within a restricted ansatz family characterized by a set of variational parameters. These parameters are optimized to provide the best possible approximation to the target state. Among the most widely used frameworks are tensor-network-based methods, including matrix product state (MPS) and matrix product operator (MPO) \cite{eisert2010colloquium,noh2020efficient,qin2024quantum,qin2024sample}, as well as projected entangled pair state (PEPS) and projected entangled pair operator (PEPO) \cite{alhambra2021locally,cirac2021matrix,anshu2024survey,qin2026quantum}. These representations exploit the entanglement structure of quantum states and have proven highly effective for low-dimensional systems or states with limited entanglement. However, tensor-network approaches typically encounter significant challenges when applied to systems with strong long-range correlations or highly entangled states, where the required bond dimensions grow rapidly. An alternative paradigm for simulating quantum many-body systems is provided by quantum Monte Carlo (QMC) methods \cite{ceperley1986quantum,foulkes2001quantum,becca2017quantum}. Despite their success in many contexts, QMC algorithms often suffer from the sign problem \cite{troyer2005computational,pan2022sign}, which leads to an exponential scaling of the computational cost. Moreover, their convergence can become extremely slow for large systems near criticality or in other challenging regimes of statistical physics \cite{carrasquilla2020machine}.

Neural networks have emerged as powerful variational representations for quantum many-body systems, giving rise to the framework of neural quantum states (NQS). These models provide compact and expressive parameterizations capable of capturing the intricate correlations present in physically relevant quantum states. A broad range of neural architectures has been explored for this purpose, including restricted Boltzmann machines \cite{torlai2018neural,torlai2019integrating}, feedforward neural networks \cite{cai2018approximating}, recurrent neural networks \cite{morawetz2021u,iouchtchenko2023neural}, convolutional neural networks \cite{fu2024lattice}, variational autoencoders \cite{rocchetto2018learning}, and Transformer-based models \cite{cha2021attention,ma2023tomography}. These developments demonstrate that neural networks can effectively exploit structural properties of quantum states, motivating the exploration of increasingly expressive architectures.

Among modern neural architectures, Transformer models \cite{vaswani2017attention} provide a particularly powerful framework for modeling complex high-dimensional dependencies. By relying on self-attention mechanisms, Transformers explicitly capture interactions between all elements of an input sequence, enabling flexible representations of long-range and nonlocal correlations. Owing to these capabilities, Transformer architectures have achieved remarkable success across a wide range of applications, including natural language processing \cite{radford2019language,brown2020language}, recommendation systems \cite{zhou2018deep,chen2019behavior}, reinforcement learning \cite{chen2021decision,janner2021offline}, computer vision \cite{dosovitskiy2020image}, and multimodal learning \cite{tsai2019multimodal}. These properties also make Transformers particularly attractive for representing quantum many-body states, where long-range correlations play a central role. Consequently, Transformer-based approaches have recently been developed for representing quantum states in both discrete configuration spaces \cite{luo2022autoregressive,zhang2023transformer,wu2023nnqs,shang2023solving,viteritti2023transformer,roca2024transformer,rende2025queries,chen2025convolutional,shang2025solving} and continuous coordinate spaces \cite{ruggeri2018nonlinear,han2019solving,pfau2020ab,hermann2020deep,adams2021variational,von2022self,li2022fermionic,pescia2022neural,hermann2023ab,medvidovic2024neural}. Beyond wavefunction representations, Transformer models have also been applied to related quantum information tasks such as quantum chemistry simulations \cite{von2022self,shang2023solving,wu2023nnqs} and quantum state tomography \cite{cha2021attention,zhong2022quantum,ma2025tomography}.

A central capability of neural networks is their ability to generalize to unseen data that share structural similarities with the training distribution. In the context of neural quantum states, the generalization performance of Transformer-based architectures has been empirically demonstrated in \cite{zhang2023transformer}. Despite these promising results, a theoretical understanding of the generalization properties of Transformer-based neural quantum states remains largely absent. In particular, it is unclear how to quantitatively characterize the associated generalization error. A distinctive property of Transformer models is their ability to perform in-context learning (ICL) \cite{min2021metaicl}. Through ICL, a model can adapt to previously unseen tasks at inference time by conditioning on a sequence of input-output examples, without updating its parameters. Compared with conventional inference procedures, this mechanism allows the model to exploit a small amount of task-specific data during testing, potentially improving generalization performance. Although recent theoretical works have begun to analyze the mechanisms underlying ICL \cite{garg2022can,akyurek2022learning,von2023transformers,zhang2024trained,huang2023context,chen2024training,yang2024context,zhang2025training,mahankali2023one,ahn2023transformers,li2024fine,li2025gating,fu2024transformers,li2025transformers,qin2026learning}, these results do not directly yield a theoretical characterization of the generalization behavior of Transformer-based neural quantum states. Addressing this gap naturally leads to the following fundamental question:

\vspace{0.2cm}

\smallskip
\begin{mdframed}[linewidth=0.3pt, leftmargin=0mm, rightmargin=0mm,
  innerleftmargin=1mm, innerrightmargin=1mm, innertopmargin=1mm, innerbottommargin=1mm]
\centering {\bf Question}: How can we theoretically characterize the generalization error of Transformer-based neural quantum states under in-context learning?
\end{mdframed}

\vspace{0.2cm}

In this paper, we provide a theoretical characterization of the generalization behavior of Transformer-based neural quantum states under ICL by establishing a rigorous inference-time generalization error bound in terms of mean squared error (MSE). We begin by analyzing the pointwise generalization error of Transformer-based neural quantum states without physical constraints, focusing on its dependence on both the number of in-context examples and the Transformer depth. Specifically, we establish the existence of a Transformer architecture whose pointwise prediction error scales inversely with both the number of in-context examples and the network depth, thereby revealing a trade-off between architectural depth and the amount of contextual information available at inference time. Notably, we show that the required Transformer depth scales linearly with the input domain dimension $nD$ (for the continuous setting with $n$ particles in $\R^D$) or $nd$ (for the discrete setting with $n$ qudits of local dimension $d$). Building on this pointwise result, we extend our analysis to the full quantum state setting, formulated in terms of rank-one density operators, and establish MSE-based generalization bounds over both continuous and discrete domains for Transformer-based neural quantum states under physical constraints. These bounds preserve the structure of the pointwise analysis, with the additional scaling factors absorbed into the normalization of the MSE over the continuous and discrete quantum state spaces. Finally, numerical simulations corroborate our theoretical predictions.

We next briefly comment on the proof strategy and its relation to prior work. Our analysis builds on the ICL framework of \cite{li2025transformers}, but is developed for the setting of Transformer-based neural quantum states, where the goal is to characterize inference-time generalization of quantum wavefunctions/states and their induced density operators. Compared with \cite{li2025transformers}, our setting introduces several structural changes, including relaxed moment  assumptions, general bounded-energy continuous domains, and discrete many-body quantum configurations, as well as coupled predictions of normalized quantum wavefunctions/states and their induced density operators. These extensions require non-trivial modifications of the original arguments, while retaining part of the underlying analytical framework. For clarity, the detailed technical developments are deferred to \Cref{subsection of physical representation of NQS} and Appendix~\ref{sec: proof of main theorem}.

{\bf Notation}: We use bold capital letters (e.g., $\mY$) to denote matrices,  bold lowercase letters (e.g., $\vy$) to denote vectors, and italic letters (e.g., $y$) to denote scalar quantities.  Elements of matrices are denoted in parentheses, as in Matlab notation. For example, $\mY(s_1, s_2)$ denotes the element in position $(s_1, s_2)$ of the matrix $\mY$. $\|\vx\|_2$ and $\|\vx\|_1$ denote the $\ell_2$- and $\ell_1$- norm of $\vx$. For a positive integer $K$, $[K]$ denotes the set $\{1,\dots, K \}$.  For two positive quantities $a,b\in \real$,  $b = O(a)$ and $b\lesssim a$ mean $b\leq c a$ for some universal constant $c$; likewise, $b = \Omega(a)$ and $b \gtrsim a$  represent $b\ge ca$ for some universal constant $c$.

\section{Basic Concepts for Transformer-Based Neural Quantum States}

\subsection{Neural Quantum States}

In this section, we introduce the basic formulation of neural quantum states (NQS). We begin with quantum systems possessing continuous degrees of freedom. A continuous $n$-particle quantum wavefunction is defined as
\begin{eqnarray}
\label{eq:continuous_wavefunction}
\psi(\vx_1,\vx_2,\dots,\vx_n)  : (\R^{D})^n\to \C,
\end{eqnarray}
where $\vx_i \in \R^D$ denotes the position (or another continuous degree of freedom) of the $i$-th particle. Consequently, the total input dimension of the wavefunction is $nD$. This formulation encompasses a wide range of physical systems, including atoms in three-dimensional space ($D=3$), electrons in two-dimensional materials such as graphene ($D=2$), and particles confined to one-dimensional chains or quantum wires ($D=1$). The wavefunction outputs a complex-valued scalar that represents the probability amplitude associated with a given configuration $(\vx_1,\dots,\vx_n)$. To ensure a valid probabilistic interpretation, the wavefunction is normalized according to
\begin{eqnarray}
\label{eq:continuous_wavefunction normalization}
\int_{\R^D}\!\cdots\!\int_{\R^D}
|\psi(\vx_1,\dots,\vx_n)|^2
\, d\vx_1 \cdots d\vx_n = 1.
\end{eqnarray}

We next consider quantum systems with discrete degrees of freedom. In this case, a quantum state is described by a wavefunction
\begin{eqnarray}
\label{eq:discrete_wavefunction}
\psi(x_1,x_2,\dots,x_n)  : (\{0,\dots, d-1\})^n\to \C,
\end{eqnarray}
where $x_i$ labels the local basis state of the $i$-th site, and $d$ denotes the local Hilbert space dimension. For example, $d=2$ corresponds to spin systems or qubits, while larger values of $d$ arise in higher-spin or multi-level systems. The choice of the computational basis $\{0,\dots,d-1\}$ is made for notational convenience; alternative local bases can be adopted without loss of generality. The normalization condition in the discrete setting is given by
\begin{eqnarray}
\label{eq:discrete_wavefunction normalization}
\sum_{x_1}\cdots\sum_{x_n}
|\psi(x_1,\dots,x_n)|^2 = 1.
\end{eqnarray}
To streamline notation, we will use the unified expression $\psi(\vx_1,\dots,\vx_n)$ throughout the remainder of the paper. Whenever it is necessary to distinguish between continuous and discrete degrees of freedom, we will explicitly specify the corresponding configuration space and measure.

The central idea of NQSs is to employ neural networks as variational ans\"atze for quantum wavefunctions. Specifically, a neural network parameterized by $\mTheta$ is used to represent a complex-valued wavefunction $\psi_{\mTheta}(\vx_1,\dots,\vx_n)$ which maps a configuration $(\vx_1,\dots,\vx_n)$ to the corresponding probability amplitude. In this framework, the network input consists of configurations $(\vx_1,\dots,\vx_n)$, while the output is taken to be the real and imaginary components of the wavefunction, $[ \Re{\psi_{\mTheta}(\vx_1,\dots,\vx_n)} \ \Im{\psi_{\mTheta}(\vx_1,\dots,\vx_n)}  ]^\top$. In this work, we adopt the real/imaginary decomposition for analytical simplicity. We note that the choice of representation is not merely a matter of convention, as different parameterizations may induce distinct inductive biases depending on the symmetry structure of the target wavefunction. In particular, certain symmetry classes may be more naturally aligned with either the Cartesian (real/imaginary) or amplitude-phase parameterization, which can affect both representation efficiency and optimization behavior. A systematic study of these effects across different symmetry regimes is left to future work.

The feasibility of representing quantum states using neural networks stems from their expressive power. Classical universal approximation results state that neural networks with at least one hidden layer, sufficiently many parameters, and a nonlinear activation function can approximate arbitrary continuous functions to any desired accuracy \cite{barron1993universal}. This expressive capability makes neural networks promising candidates for constructing variational wavefunctions $\psi_{\mTheta}(\vx_1,\dots,\vx_n)$ that closely approximate the exact quantum state $\psi(\vx_1,\dots,\vx_n)$. A natural question that follows is which neural network architectures are most suitable for this task.

\subsection{Transformer}

A central challenge in neural quantum state modeling is to efficiently represent complex many-body wavefunctions with strong, nonlocal correlations among their degrees of freedom. In particular, quantum systems often exhibit long-range and highly structured dependencies that are difficult to capture using architectures based solely on local connectivity.  In this work, we adopt the Transformer architecture as the backbone of our neural quantum state models. The Transformer is a sequence-processing architecture built upon attention mechanisms and feed-forward networks, and is especially well suited for modeling high-dimensional dependencies and long-range correlations. Its attention mechanism enables direct interactions between arbitrary pairs of input elements, providing a natural inductive bias for representing correlated quantum configurations. These properties make Transformers a compelling choice for approximating complex quantum many-body wavefunctions.

We now briefly review the Transformer architecture. We begin by introducing the attention layer, which constitutes the core computational primitive of the Transformer, and then describe the feed-forward layer. Finally, we explain how these components are composed into a multi-layer Transformer through residual connections.

\begin{itemize}

\item{\textbf{Attention layer.} \ } Let $\mX \in \R^{E \times (N+1)}$ denote the input matrix, where $E$ is the embedding dimension and $N+1$ is the sequence length. The (self-)attention operator is defined as
\begin{eqnarray}
    \label{The definition of attention layer}
    \textup{attn}(\mX; \mQ, \mK, \mV ) = \frac{1}{N} \mV \mX \sigma_\textup{attn}(\mX^\top \mQ^\top \mK \mX  ),
\end{eqnarray}
where $\mQ, \mK, \mV \in \R^{E \times E}$ are learnable parameter matrices referred to as the query, key, and value matrices, respectively. This operator aggregates information across all input positions by computing interaction weights between embeddings, thereby enabling the model to capture global dependencies. The activation function $\sigma_\textup{attn}$ is chosen as the sigmoid function
\begin{eqnarray}
    \label{definition of sigmoid function}
    \sigma_\textup{attn}(x ) = \frac{e^x}{e^x +1},
\end{eqnarray}
which is applied entrywise. The normalization factor $1/N$ ensures appropriate scaling with respect to the sequence length.

To further enhance expressivity, the Transformer employs a multi-head attention mechanism. Instead of computing a single attention map, multiple attention heads are applied in parallel, each with its own set of parameters. Denoting by $\textup{Attn}_{\mTheta_\textup{attn}}(\cdot)$ a multi-head attention layer parameterized by
\[
\mTheta_\textup{attn} = \{ \mQ_m, \mK_m, \mV_m \}_{m=1}^M \subset \R^{E \times E},
\]
the output of the attention layer with $M$ heads is given by
\begin{eqnarray}
    \label{definition of attention layer output with multiheads}
    \textup{Attn}_{\mTheta_\textup{attn}}(\mX) = \mX + \sum_{m=1}^M \textup{attn}(\mX; \mQ_m, \mK_m, \mV_m ).
\end{eqnarray}
The residual connection preserves the original input while allowing each attention head to contribute additional features. This design improves optimization stability and facilitates the propagation of information across layers.

\item{\textbf{Feed-forward layer.} \ } While the attention layer is responsible for modeling interactions across different positions in the sequence, the feed-forward layer applies a nonlinear transformation independently to each embedding. Given an input matrix $\mX$, the feed-forward layer is defined as
\begin{eqnarray}
    \label{definition of MLP layer}
    \textup{FF}_{\mTheta_\textup{ff}}(\mX) = \mX + \mU\sigma_{\textup{ff}}(\mW \mX),
\end{eqnarray}
where $\mTheta_\textup{ff} = \{ \mU, \mW \} \in \R^{E \times E}$ denotes the learnable weight matrices. The activation function $\sigma_{\textup{ff}}$ is chosen to be the ReLU function,
\begin{eqnarray}
    \label{definition of Relu function}
    \sigma_{\textup{ff}}(x ) = x \mathbf{1}_{x>0},
\end{eqnarray}
which is applied entrywise. As with the attention layer, a residual connection is incorporated, ensuring that the feed-forward transformation refines rather than overwrites the existing representation.

\end{itemize}

\paragraph{Multi-layer transformers.}
By stacking attention and feed-forward layers, the Transformer builds hierarchical representations of the input sequence. Starting from the initial embedding matrix $\mH^{(0)}  \in \R^{E \times (N+1)}$, which is specified a priori, a Transformer with $L$ layers computes successive hidden representations according to
\begin{eqnarray}
    \label{definition of multi-layer transformer}
    \mH^{(\ell)}=\textup{FF}_{\mTheta_\textup{ff}^{(\ell)}}( \textup{Attn}_{\mTheta_\textup{attn}^{(\ell)}}  ( \mH^{(\ell -1)} )), \ell = 1,\dots, L,
\end{eqnarray}
with the final output given by
\begin{eqnarray}
    \label{the output of multi-layer transformer}
    \textup{TF}_{\mTheta}(\mH^{(0)}) \coloneqq \mH^{(L)}.
\end{eqnarray}
Here, $\mTheta$ collects all parameters across layers,
\begin{eqnarray}
    \label{all transformer parameters}
    \mTheta = \{ \{ \mQ_m^{(\ell)}, \mK_m^{(\ell)}, \mV_m^{(\ell)} \}_{m=1}^M, \mU^{(\ell)}, \mW^{(\ell)}    \}, \ \ell = 1,\dots, L.
\end{eqnarray}
This layered structure allows the model to progressively refine its internal representation, capturing increasingly complex dependencies among the input components. In addition,we note that the parameters $\mTheta$ of the Transformer are learned during an offline training phase by minimizing a task-dependent objective function defined over a collection of training configurations. In this work, we adopt a first-order optimization scheme and update $\mTheta$ using the Adam optimizer~\cite{kingma2014adam}, which adaptively adjusts the step size for each parameter. Once trained, the resulting Transformer defines a fixed neural quantum state representation, which is subsequently evaluated at test time.

Finally, the Transformer produces its prediction by applying a readout operation to the output of the final layer. In particular, the prediction $\wh \vy_{N+1} = \begin{bmatrix} \wh \vy_{N+1}(1) & \wh \vy_{N+1}(2)\end{bmatrix}^\top\in\R^{2}$ corresponding to the $(N+1)$-th input position is extracted from the last column of $\textup{TF}_{\mTheta}(\mH^{(0)})$ as
\begin{equation}
\label{output of transformer neural quantum states}
\begin{split}
    &\wh \vy_{N+1}(1)  =  \textup{ReadOut}(\textup{TF}_{\mTheta}( \mH^{(0)})) = [ \mH^{(L)} ]_{E-1,N+1}, \\
    &\wh \vy_{N+1}(2)  =  \textup{ReadOut}(\textup{TF}_{\mTheta}( \mH^{(0)})) = [ \mH^{(L)} ]_{E,N+1}.
\end{split}
\end{equation}

\section{Generalization of Transformer-Based Neural Quantum States}

\subsection{In-context Learning}

Once a Transformer has been trained to represent a neural quantum state from a given class of quantum wavefunctions, a natural next question is how to assess its generalization ability to unseen configurations that nevertheless share similar structural properties with the training data. In this section, we introduce ICL as a principled framework for evaluating such generalization behavior.

We emphasize that the ICL mechanisms studied in this work operate exclusively during inference and do not involve any modification of the trained parameters $\mTheta$. Instead, ICL probes the model's ability to utilize contextual information provided at test time. Specifically, given a set of input-output pairs $\{(\vx_{i,1}, \dots,\vx_{i,n}), \psi(\vx_{i,1}, \dots,\vx_{i,n})\}_{i=1}^{N}$ serving as context, the objective of ICL is to use the fixed Transformer parameters $\mTheta$ to predict the wavefunction value $\psi((\vx_{N+1,1}, \dots,\vx_{N+1,n}))$ for a new query configuration $(\vx_{N+1,1}, \dots,\vx_{N+1,n})$. To enable the application of the Transformer to complex-valued wavefunctions, we represent each data point $\psi(\vx_{i,1}, \dots,\vx_{i,n})$ as a two-dimensional real vector encoding its real and imaginary components. Concretely, we define
\begin{equation}
\label{each element in the d vector form}
\vd(\{\vx_{i,j}\}_{j=1}^{n} ) = \begin{bmatrix}
d_1(\vx_{i,1}, \dots,\vx_{i,n}) \\ d_2(\vx_{i,1}, \dots,\vx_{i,n})
\end{bmatrix} =
\begin{bmatrix}
\Re{\psi(\vx_{i,1}, \dots,\vx_{i,n})} \\ \Im{\psi(\vx_{i,1}, \dots,\vx_{i,n})}
\end{bmatrix}
\in\R^{2}.
\end{equation}
An illustration of the ICL framework is shown in Fig.~\ref{fig:summary_ICL}. Given a sequence of in-context examples followed by a query instance, the inputs are first embedded and combined with positional encodings. The resulting representations are then iteratively refined through multiple Transformer layers consisting of masked multi-head self-attention and feedforward layers with residual connections. The final hidden state corresponding to the query token is mapped back to the output space via an unembedding layer.
\begin{figure*}[!ht]
  \centering
  \includegraphics[width=12.5cm]{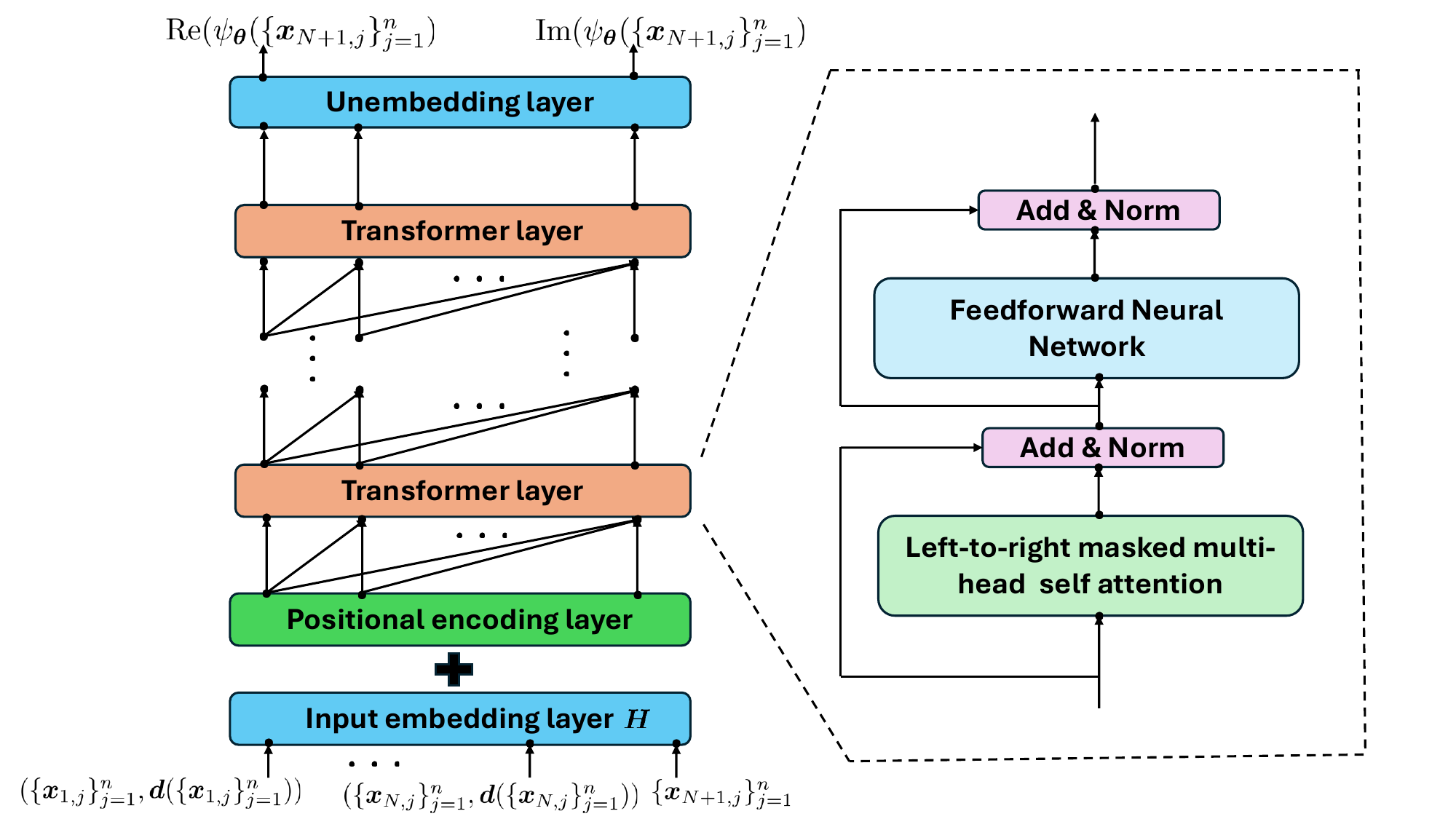}
  \caption{In-Context Learning Framework for Transformer-Based Neural Quantum States.}
  \label{fig:summary_ICL}
\end{figure*}

Finally, we introduce the embedding matrix $\mH\in\R^{E\times(N+1)}$ as follows:
\begin{eqnarray}
    \label{the construction of embedding vector}
    \mH := \begin{bmatrix}\vx_{1,1} & \cdots & \vx_{N,1} & \vx_{N+1,1} \\ \vdots &   & \vdots & \vdots  \\ \vx_{1,n} & \cdots & \vx_{N,n} & \vx_{N+1,n}  \\ \vd(\{\vx_{1,j}\}_{j=1}^{n} ) & \cdots & \vd(\{\vx_{N,j}\}_{j=1}^{n} ) & {\bm 0}  \end{bmatrix}\in \R^{E \times (N+1)},
\end{eqnarray}
where the first $N$ columns correspond to training samples, while the last column represents the query configuration. Since the target quantity associated with the query is unavailable during inference, the corresponding coordinates are initialized to zero.

\subsection{Generalization Analysis for Transformer-Based Neural Quantum States}
\label{subsection: pointwise generalization of NQS}

In this subsection, we formalize the in-context prediction procedure as a statistical learning problem and derive explicit generalization bounds for Transformer-based prediction of neural quantum states. As shown in \eqref{output of transformer neural quantum states}, during inference the Transformer defines a deterministic mapping
\begin{eqnarray}
\label{output of transformer neural quantum states true}
\widehat{\vy}_{N+1} = \textup{ReadOut}(\textup{TF}_{\mTheta}( \mH)),
\end{eqnarray}
where $\mH$ encodes the in-context sequence and the query configuration. For the theoretical analysis, we consider an augmented embedding matrix $\mH$ defined in \eqref{the construction of embedding vector each layer} of Appendix~\ref{Proof of transformer emulate proximal GD}. This formulation retains the original embedding structure in \eqref{the construction of embedding vector} while incorporating additional components required for the theoretical characterization. These auxiliary coordinates are introduced solely for the constructive proof of \Cref{thm:generalization_transformer_nqs}. They encode intermediate quantities arising in the layer-by-layer simulation of the optimization procedure, thereby enabling an explicit characterization of the Transformer's computation and facilitating the theoretical analysis. They do not provide additional information about the underlying quantum wavefunction/state and are not required in practical implementations. In addition, since the model parameters $\mTheta$ are fixed at inference, the generalization error depends on the context size $N$ and the structure of the input domain.

Our objective is to rigorously characterize the scaling of the expected prediction error as a function of $(i)$ the number of in-context examples $N$, $(ii)$ the architectural depth $L$ of the Transformer, and $(iii)$ the structural complexity of the input domain.
\begin{theorem}
\label{thm:generalization_transformer_nqs}
Suppose that the in-context samples and the query input are drawn i.i.d. from an underlying distribution supported on either the continuous domain $\{ (\vx_1,\dots,\vx_n): \sum_{i=1}^n \|\vx_i\|_2^2 \le B, \vx_i\in\R^{D}, i\in[n] \}$ or the discrete domain  $\{ (x_1,\dots,x_n): x_i \in \{0,\dots,d-1\},\ i \in [n]\}$ and that each component of $\vd(\cdot)$ belongs to a function class $\mathcal{G}$ whose elements have bounded first absolute moment on the corresponding domain.
Denote the transformer prediction on the query by $\widehat{\vy}_{N+1} \in \R^2$. Consider a transformer with $L$ layers. Let $M$ denote the number of attention heads, with $M\geq 8$. Assume further that
\begin{eqnarray}
    \label{Upper bound of the layers in the Transformer}
    \begin{cases}
    L\geq B, \ \frac{L}{\log L} \geq \frac{2nDB\sqrt{N}}{\sqrt{\log N}} & \textup{continuous domain},   \\
    L\ge d, \ \ \frac{L}{\log L} \geq \frac{2nd\sqrt{N}}{\sqrt{\log N}}, & \textup{discrete domain}.
    \end{cases}
\end{eqnarray}
Then, with high probability, the expected prediction error satisfies
\begin{eqnarray}
\label{Upper bound of the estimated output and ground truth output main paper simpified}
\E\left[\left\|\wh \vy_{N+1}- \vd(\{\vx_{N+1,j}\}_{j=1}^{n})\right\|_2^2\right] = O\left(\sqrt{\frac{\log N}{N}}\right),
\end{eqnarray}
where the input configuration $\{\vx_{N+1,j}\}_{j=1}^{n}$ is defined over either a continuous or discrete domain and and the constant hidden in the $O(\cdot)$ notation depends on the function class $\mathcal{G}$.
\end{theorem}
The complete derivation is given in {Appendix}~\ref{sec: proof of main theorem}. The leading term $\sqrt{(\log N)/N}$ characterizes the estimation error induced by the finite number of in-context samples. This implies that increasing the context size $N$ directly improves generalization performance.  In addition, the depth condition in~\eqref{Upper bound of the layers in the Transformer} is a sufficient condition to guarantee the simplified generalization bound in~\eqref{Upper bound of the estimated output and ground truth output main paper simpified}. In the continuous and discrete settings, the sufficient depth prescribed by Theorem~\ref{thm:generalization_transformer_nqs} scales as $nDB\log(BL) \sqrt{N}/\sqrt{\log N}  $ and $nd\log(dL) \sqrt{N}/\sqrt{\log N}$, respectively.
When $D$, $d$ and $N$ are treated as fixed constants, the dominant dependence is linear in the input domain dimension $nD$ or $nd$ up to logarithmic factors in $L$. In the continuous case, the dependence on the energy constraint parameter $B$ reflects the growth of the effective input domain. A larger $B$ corresponds to a larger admissible region of quantum wavefunctions, which in turn requires increased model capacity--here quantified by depth $L$--to maintain a fixed approximation accuracy. Moreover, we emphasize that $M\geq 8$ is a sufficient condition for the above result. The proof in {Appendix}~\ref{sec: proof of main theorem} provides a more precise layer-wise attention-head configuration that is sufficient for constructing the Transformer used in our analysis.

To further clarify the dependence among the model parameters, we recall the original generalization error bounds derived in Appendix~\ref{sec: Composition of Approximations} before imposing the sufficient depth condition:
\begin{eqnarray}
\label{Upper bound of the estimated output and ground truth output main paper original}
&\!\!\!\!\!\!\!\!&\E\left[\left\|\wh \vy_{N+1}- \vd(\{\vx_{N+1,j}\}_{j=1}^{n})\right\|_2^2\right] \nonumber\\
&\!\!\!\!=\!\!\!\!&     \begin{cases}
    O\bigg(\sqrt{\frac{\log N}{N}} + \frac{nDB\max\{\log(BL), \log(B\sqrt{N}/\sqrt{\log N}) \}}{L} \bigg), & \textup{continuous domain},   \\
    O\bigg(\sqrt{\frac{\log N}{N}} + \frac{nd\max\{\log(dL), \log(d\sqrt{N}/\sqrt{\log N}) \}}{L} \bigg), & \textup{discrete domain}.
    \end{cases}
\end{eqnarray}
Here, the depth condition in~\eqref{Upper bound of the layers in the Transformer} ensures that the approximation error term is dominated by the statistical error term, thereby yielding the simplified rate in~\eqref{Upper bound of the estimated output and ground truth output main paper simpified}. This expression also reveals that the depth-dependent complexity enters through a factor of order $\log(L)/L$, up to additional logarithmic terms.

Finally, we emphasize that the theorem provides an existence guarantee: it establishes that there exists a Transformer architecture achieving the stated bound, rather than asserting uniform optimality over all parameter choices. Nevertheless, the result reveals a fundamental scaling law governing the interplay between context size, architectural depth, and domain complexity. In this sense, it provides theoretical insight into the capacity requirements of Transformer-based neural quantum states.
In addition, the augmented embedding matrix $\mH$ considered in this theorem, defined in \eqref{the construction of embedding vector each layer} of Appendix~\ref{Proof of transformer emulate proximal GD}, is introduced to facilitate the theoretical analysis. This formulation retains the essential structure of the practical embedding in \eqref{the construction of embedding vector} while incorporating additional components required for the theoretical characterization. In the subsequent simulations, we employ the practical embedding construction in \eqref{the construction of embedding vector} and observe the same scaling behaviors predicted by \Cref{thm:generalization_transformer_nqs}.

\paragraph{Proof Sketch}

The proof is constructive: we explicitly construct an in-context prediction rule realizable by a depth-$L$ Transformer and establish a quantitative bound on its prediction error.

We begin by rewriting the complex-valued wavefunction as two real-valued components $d_1(\cdot)$ and $d_2(\cdot)$. For $k \in \{1,2\}$, we introduce a feature map $\vphi_k^{N+1} \in \mathbb{R}^{P_k+1}$ depending on the query configuration. We consider a linear-in-features representation of $d_k(\cdot)$ of the form ${\vphi_k^{N+1}}^\top \va_k$, where $\va_k$ denotes the associated coefficient vector, and the resulting approximation error is controlled by the truncation of the feature expansion. The coefficient vector $\va_k$ is estimated from the $N$ in-context samples via the Lasso estimator \cite{tibshirani1996regression}, with a regularization parameter that that depends on $d_k(\cdot)$ and is specified in \Cref{Generalization error bound of Lasso estimator} in Appendix~\ref{sec:Approximation of the Target Function Class by the Lasso Estimator}. This yields the estimator $\wh{\va}_k$.  Combining the above linear-in-features representation with standard Lasso estimation guarantees for linear models, we obtain that with high probability the prediction error satisfies
\begin{eqnarray}
    \label{upper bound of first term in the final conclusion}
    \E[({\vphi_{k}^{N+1}}^\top \wh\va_{k} -  d_k(\vx_{i,1}, \dots,\vx_{i,n}) )^2]\leq \begin{cases}
     C_1\sqrt{\frac{\log N}{N}} + \frac{C_2nDB\log(BL)}{L}, & \textup{continuous},\\
     C_3\sqrt{\frac{\log N}{N}} + \frac{C_4nd\log(dL)}{L}, & \textup{discrete},
    \end{cases}
\end{eqnarray}
for some positive constants $C_1,\dots,C_4$.

We then connect the statistical estimator to the Transformer architecture by showing that the Lasso solution can be approximately realized by a depth-$L$ Transformer. While the prediction error of the Lasso estimator has been characterized in~\eqref{upper bound of first term in the final conclusion}, the associated optimization problem generally does not admit a closed-form solution. Consequently, we consider its iterative realization via the inexact proximal gradient method. The proximal gradient formulation serves as a computational bridge between the statistical estimator and the Transformer architecture. Specifically, we prove that one Transformer layer can emulate one proximal gradient iteration up to controlled error. Therefore, stacking $L$ Transformer layers yields an approximate implementation of the iterative optimization procedure, producing an output $\wh\vy_{N+1}$ satisfying
\begin{eqnarray}
    \label{upper bound of second term in the final conclusion}
    ({\vphi_{k}^{N+1}}^\top \wh\va_{k} -  \wh \vy_{N+1}(k) )^2\leq \frac{C_5}{N^2}
\end{eqnarray}
for both the continuous and discrete domains, where $C_5>0$ is a constant.

Combining these estimates and applying a error decomposition yields the bound in \eqref{Upper bound of the estimated output and ground truth output main paper original}. Conceptually, the argument reduces the nonlinear prediction problem to a finite-dimensional linear estimation problem through the induced feature representation.  This reduction makes the dependence on the problem parameters explicit: once the target function is represented through a controlled feature map, the generalization error can be analyzed using regression theory. As a result, the bound reveals how the prediction error scales with the context size $N$, the system size $n$, the domain parameters $(D,d,B)$, and the architectural depth $L$.

\subsection{Physical Representation of Transformer-Based Neural Quantum States}
\label{subsection of physical representation of NQS}

In the preceding section, we investigated the generalization properties of Transformer-based neural quantum states. However, the resulting representations are not necessarily normalized, whereas physically meaningful quantum wavefunctions and quantum states must satisfy the normalization conditions in \eqref{eq:continuous_wavefunction normalization} and \eqref{eq:discrete_wavefunction normalization}. Such normalization constraints are fundamental to ensuring the physical admissibility of the corresponding quantum representations. Motivated by this consideration, we extend the preceding analysis from general Transformer-based neural quantum states to physically admissible quantum wavefunctions and quantum states, and subsequently generalize the framework to encompass the representation of physical density operators. To this end, we define the predicted normalized continuous quantum wavefunction $\ol\psi_{\textup{c}}$ satisfying\footnote{In practical implementations, the normalization integral in \eqref{eq:continuous_wavefunction normalization transformer1} is generally intractable in high-dimensional continuous domains. Consequently, the normalization condition is approximated through Monte Carlo sampling. More precisely, given sampling points $\{(\vx_1^{(k)},\dots,\vx_n^{(k)})\}_{k=1}^K$ drawn from a reference distribution $q$, the normalization condition can be approximated by
\begin{eqnarray}
\label{eq:continuous_wavefunction normalization transformer}
\frac{1}{K}\sum_{k=1}^K
\frac{
|\ol\psi_{\textup{c}}(\vx_1^{(k)},\dots,\vx_n^{(k)})|^2
}{
q(\vx_1^{(k)},\dots,\vx_n^{(k)})
}
\approx 1.
\end{eqnarray}}
\begin{eqnarray}
\label{eq:continuous_wavefunction normalization transformer1}
\int_{\sum_{i=1}^n \|\vx_i\|_2^2 \le B}
|\ol\psi_{\textup{c}}(\vx_1,\dots,\vx_n)|^2
\, d\vx_1 \cdots d\vx_n = 1
\end{eqnarray}
and similarly define the predicted normalized discrete quantum state $\ol\psi_{\textup{d}}$ satisfying\footnote{For discrete quantum states, the normalization condition can be computed explicitly by summing over the discrete configuration space.}
\begin{eqnarray}
\label{eq:discrete_wavefunction normalization transformer}
\sum_{x_1}\cdots\sum_{x_n}
|\ol\psi_{\textup{d}}(x_1,\dots,x_n)|^2 = 1.
\end{eqnarray}
Building upon the above physically normalized representations, we next present the corresponding generalization analysis for Transformer-based neural quantum states.
\begin{theorem}
\label{thm:generalization normalized_transformer_nqs}
Under the same setting as in \Cref{thm:generalization_transformer_nqs}, let $\psi^\star$ and $\ol\psi$ denote the ground-truth quantum state and the corresponding normalized prediction, respectively. Define
\begin{eqnarray}
\label{definition of alpha}
\alpha=
\begin{cases}
\left(\dfrac{2\pi e B}{nD}\right)^{\frac{nD}{2}}, & \textup{continuous domain},\\[2mm]
d^n, & \textup{discrete domain},
\end{cases}
\end{eqnarray}
and let $\|\cdot\|$ denote  the $L^2$ norm in the continuous domain and the $\ell_2$ norm in the discrete domain\footnote{The $L^2$ norm and $\ell_2$ norm are defined as $ \|\psi_{\mathrm{c}}\|_{L^2}^2 = \int_{\sum_{i=1}^n \|\mathbf{x}_i\|_2^2 \le B}
|\psi_{\mathrm{c}}(\vx_1,\dots,\vx_n)|^2\, d\vx_1 \cdots d\vx_n$ and $ \|\psi_{\mathrm{d}}\|_2^2 = \sum_{x_1,\dots,x_n} |\psi_{\mathrm{d}}(x_1,\dots,x_n)|^2$.}. Consider a transformer with $L$ layers. Let $M$ denote the number of attention heads, with $M\geq 8$. Assume further that
\begin{eqnarray}
    \label{Upper bound of the layers in the Transformer another}
    \begin{cases}
    L\geq B, \ \frac{L}{\log L} \geq \frac{2nDB\sqrt{N}}{\sqrt{\log N}} & \textup{continuous domain},   \\
    L\ge d, \ \ \frac{L}{\log L} \geq \frac{2nd\sqrt{N}}{\sqrt{\log N}}, & \textup{discrete domain}.
    \end{cases}
\end{eqnarray}
Then, with high probability, the following generalization bounds hold:
\begin{eqnarray}
    \label{Upper bound of the estimated output and ground truth output main paper normalized simpified}
    \frac{\E[\|\ol\psi-\psi^\star\|^2]}{\alpha} = O\!\left(\sqrt{\frac{\log N}{N}}\right).
\end{eqnarray}
Moreover, these bounds immediately imply corresponding guarantees for the induced density operators\footnote{The Frobenius norms are defined as $\|\psi_{\mathrm{c}}\psi_{\mathrm{c}}^\dagger\|_F^2
=
\int_{\sum_{i=1}^n \|\vx_i\|_2^2 \le B}\!\!\int_{\sum_{i=1}^n \|\vx_i'\|_2^2 \le B}
|\psi_{\mathrm{c}}(\vx_1,\dots,\vx_n)\psi_{\mathrm{c}}(\vx_1',\dots,\vx_n')|^2 \, d\vx_1 \cdots d\vx_nd\vx_1' \cdots d\vx_n'$ and $\| \psi_{\mathrm{d}} \psi_{\mathrm{d}}^\dagger\|_F^2 = \sum_{x_1,\dots,x_n,x_1',\dots, x_n'=1}^{d}
|\psi_{\mathrm{d}}(x_1,\dots,x_n)\psi_{\mathrm{d}}(x_1',\dots, x_n')|^2$.}.
\begin{eqnarray}
    \label{Upper bound of the estimated output and ground truth output main paper normalized density operator simpified}
    \frac{\E[\|\ol\psi \ \ol\psi^\dagger-\psi^\star{\psi^\star}^\dagger\|_F^2]}{\alpha} = O\!\left(\sqrt{\frac{\log N}{N}}\right).
\end{eqnarray}
\end{theorem}
The proof is provided in {Appendix}~\ref{Appendix of generalization of physical Transformer}. Theorem~\ref{thm:generalization normalized_transformer_nqs} establishes generalization bounds for the physically normalized quantum wavefunction/state predictions. Compared to the unnormalized setting in \Cref{thm:generalization_transformer_nqs}, the bounds are expressed in terms of normalized quantum states and associated physical density operators, where the scaling by $(\frac{2\pi e B}{nD})^{\frac{nD}{2}}$ for the continuous case and $d^n$ for the discrete case naturally accounts for the volume (or cardinality) of the underlying configuration space. Importantly, these factors do not reflect additional statistical or optimization complexity of the Transformer model. Instead, they arise from converting pointwise error control into Frobenius-norm error over the full state space, i.e., integrating (or summing) over all configurations.

Although both scaling factors originate from the size of the underlying configuration space, they exhibit seemingly different asymptotic behaviors as the system size increases. Specifically, the discrete factor $d^n$ grows exponentially with the number of particles $n$, whereas the continuous factor $(\tfrac{2\pi eB}{nD})^{nD/2}$ decreases when $B$ is fixed. This apparent discrepancy is a consequence of the domain normalization adopted in the continuous setting. Under the current formulation, the total energy budget $B$ in the continuous constraint $\sum_{i=1}^n \|\vx_i\|_2^2 \le B$ is fixed independently of the number of particles $n$. Consequently, the average admissible energy per particle decreases as $B/n$, causing the configuration-space volume factor $(\tfrac{2\pi eB}{nD})^{nD/2}$ to decay as $n$ increases. This behavior should not be interpreted as a reduction in the intrinsic complexity of the continuous problem, but rather as a consequence of shrinking the feasible domain as the system size grows. In a more physically realistic scaling, where each additional particle contributes a comparable amount of admissible volume (for example, by choosing $B=nDR^2$ with a fixed per-particle radius $R$), the volume factor becomes $(2\pi eR^2)^{nD/2}$, which grows exponentially with $n$, consistent with the discrete factor $d^n$. Therefore, the apparent discrepancy between the continuous and discrete settings originates solely from different normalizations of the configuration space, rather than from any fundamental distinction between the two formulations.

We emphasize that the above analysis is not restricted to the prediction of quantum wavefunctions or quantum states. The derivation relies primarily on the generalization properties of the Transformer predictor and therefore extends naturally to a broad class of quantum-property prediction tasks. In particular, whenever the target quantity can be represented as a deterministic function of the underlying quantum state, analogous generalization guarantees can be established for Transformer-based models predicting quantum fidelities, entanglement entropies, two-point correlation functions, expectation values of local observables, and energy variances of many-body Hamiltonians. This observation suggests that the proposed framework provides a general theoretical foundation for analyzing the generalization behavior of Transformer-based quantum learning models beyond neural quantum state representations. As an example, one may consider a Transformer-based model that directly predicts physical properties of a structured quantum state represented as a MPO. Let $\vrho^\star$ denote a quantum state admitting an MPO representation with maximum bond dimension $\ol r$, and let $f_{\vtheta}(\vrho^\star)$ denote a target quantum property. Here, $\vtheta$ indexes a family of quantum property functionals, where each choice of $\vtheta$ specifies a particular rule for extracting a scalar property from a density operator. The class of MPOs with physical dimension $d$ and bond dimension at most $\ol r$ admits a parameterization with effective complexity scaling as $O(nd^2 {\ol r}^2)$ degrees of freedom. Since $\vrho^\star$ is a unit-trace density operator, it satisfies $\|\vrho^\star\|_F \le 1$.  Denote by $\wt f_{\vtheta}(\vrho^\star)$ the corresponding prediction produced by the Transformer model. Using the same argument as in \Cref{thm:generalization_transformer_nqs}, we obtain that there exists a Transformer architecture with $L$ layers such that
\begin{eqnarray}
    \label{Upper bound of the layers in the Transformer normalized quantum property}
     L\ge d^2{\ol r}^2, \ \ \frac{L}{\log L} = \Omega\left(\frac{nd^2{\ol r}^2\sqrt{N}}{\sqrt{\log N}}  \right).
\end{eqnarray}
Then, with high probability, the expected prediction error satisfies
\begin{eqnarray}
    \label{Upper bound of the estimated output and ground truth output main paper quantum property}
    \E[|\wt f_{\vtheta}(\vrho^\star) -  f_{\vtheta}(\vrho^\star)|^2]= O\left(\sqrt{\frac{\log N}{N}}  \right).
\end{eqnarray}
Here, we note that although the in-context samples contain both quantum states and their associated quantum properties, the Transformer operates on the MPO representation, which is parameterized by tensor factors with effective degrees of freedom $O(nd^2 {\ol r}^2)$ rather than the full $d^n \times d^n$ quantum state. This reduced representation complexity explains why the architectural requirement in \eqref{Upper bound of the layers in the Transformer normalized quantum property} scales with $O(nd^2 \bar r^2)$ instead of an exponential dependence on $n$. However, the above polynomial scaling result relies on the structured MPO representation, which provides an efficient parameterization of quantum states with bond dimension $\ol r$. This assumption captures a broad class of physically relevant low-entanglement quantum states, but does not cover general quantum states with exponentially large degrees of freedom. Extending the present analysis to fully general quantum states, or identifying minimal structural conditions under which similar guarantees hold, remains an important direction for future work.

Building upon the above analysis of wavefunctions/states and the induced density operators, we now clarify the relationship between our results and prior works.
\begin{itemize}
\item \textbf{ICL-based generalization theory for quantum state learning.}
We develop a unified theoretical framework that characterizes ICL as a mechanism for generalization in Transformer-based quantum state estimation. We establish rigorous MSE-based generalization bounds at inference time, covering both continuous and discrete input domains, thereby providing a theoretical explanation of empirical observations in \cite{zhang2023transformer}. We further characterize the scaling of the generalization error with respect to the number of in-context samples, as well as the scaling of the required network depth with respect to the input domain dimension. Moreover, we extend the analysis beyond quantum state estimation to a broader class of quantum properties induced by the learned states.

\item \textbf{Extension of prior Transformer-ICL theory.}
While inspired by the proof strategy in \cite{li2025transformers}, extending their analysis to our setting is non-trivial due to several structural modifications. First, we relax the assumption that the target function is absolutely integrable and instead require only a bounded first absolute moment. This necessitates a new Fourier approximation analysis and a corresponding modification of the feature construction. Second, unlike \cite{li2025transformers}, which considers functions defined on the unit ball, our analysis is carried out over a general bounded-energy domain with radius $B$, together with discrete quantum configuration spaces. Consequently, several approximation arguments and error estimates need to be generalized to account for the enlarged geometry. Third, our Transformer simultaneously predicts two coupled outputs corresponding to the real and imaginary components of quantum wavefunctions. This requires a redesigned embedding matrix and hidden-state representation, together with a new Transformer construction that emulates the inexact proximal gradient iterations in the multi-output setting. Finally, we establish generalization guarantees for normalized quantum wavefunctions/states and the associated density operators, introducing additional normalization and operator-level error analyses that are absent from the original work. The remaining components of the proof are adapted from \cite{li2025transformers} whenever applicable.

\end{itemize}

\section{Simulation Results}
\label{sec: simulation results}

In this section, we present numerical experiments validating the generalization capability of Transformer-based neural quantum states as predicted by our theoretical analysis. To this end, we consider both continuous and discrete quantum systems, corresponding to a family of wavefunctions and a family of density matrices, respectively. For the continuous setting, we consider correlated Gaussian many-body wavefunctions with complex covariance \cite{weedbrook2012gaussian}:
$\psi(\vx_1,\dots,\vx_n) = N_{\textup{g}} e^{( -\frac{1}{2} \sum_{k,l=1}^n \vx_k^\top (\mA_{kl} - i \mB_{kl}) \vx_l)}$,
where $\mA_{kl}, \mB_{kl}\in\R^{D\times D}$ are blocks of symmetric matrices $\mA, \mB\in\R^{nD\times nD}$. Here $\mA$ is positive definite to ensure convergence, while $\mB$ is unconstrained and contributes only to the complex phase, which cancels in $|\psi|^2$. The constant $N_{\textup{g}} = \big(\frac{\det \mathbf{A}}{\pi^{nD}}\big)^{\frac{1}{4}}$ is a normalization constant.

For the discrete setting, we consider a family of thermal states $\vrho = \frac{e^{-\mH/T}}{\trace(e^{-\mH/T})}$ with temperature $T$, constructed from a disordered transverse-field Ising Hamiltonian  \cite{petHo2023random}:
$\mH = - h \sum_{k=1}^{n} \mX_k - J \sum_{k=1}^{n-1} \mZ_k \mZ_{k+1} + \sum_{k=1}^{n} g_k \mZ_k$,
where $\mX_k$ and $\mZ_k$ denote Pauli operators acting on the $k$-th qubit:
$\mX_k = \mId^{\otimes (k-1)} \otimes \mX \otimes \mId^{\otimes (n-k)}$,
$\mZ_k = \mId^{\otimes (k-1)} \otimes \mZ \otimes \mId^{\otimes (n-k)}$,
with $\mX = \begin{bmatrix} 0 & 1 \\ 1 & 0 \end{bmatrix}$ and $\mZ = \begin{bmatrix} 1 & 0 \\ 0 & -1 \end{bmatrix}$.
The parameters $J$ and $h$ control the interaction strength and transverse field, respectively, while $g_k$ are independently sampled random local fields. Although our theoretical results in \Cref{thm:generalization_transformer_nqs} are formulated for quantum states in vector form, they naturally extend to density matrices by vectorizing $\vrho$ into $\vec(\vrho)(x_1,x_2,\dots,x_n) : (\{0,\dots, d^2-1\})^n \to \C$. The only modification required is to replace $d$ with $d^2$ in the corresponding statements.

We now describe the training and evaluation protocols. During training, we adopt the AdamW optimizer with a learning rate of $10^{-3}$, which lies within the standard stable regime for Transformer optimization and is found to provide reliable and efficient training performance in our experiments. We use the mean squared error as the loss function, with the number of in-context samples fixed to $N=20$ and the number of attention heads set to $M=8$. For the continuous case, we set the matrix dimension to $E=8$ and use $2048$ training samples per epoch, where the eigenvalues of $\mA$ and $\mB$ are drawn from $[0.5, 2]$ and $[-1, 1]$, respectively.
Training inputs are sampled uniformly from the Euclidean ball $B = c \cdot \trace(\mA^{-1}) $, where $c$ is a coverage constant. Since $\trace(\mA^{-1})$ corresponds to the total variance of the Gaussian state (i.e., the sum of marginal variances), this choice scales the sampling region proportionally to the characteristic dispersion of the wavefunction. As a result, the sampled configurations adequately cover the high-probability mass of the state, allowing both amplitude and phase variations induced by $\mB$ to be sufficiently represented during training.
During evaluation, we vary the number of in-context samples $N$, and for each value, the procedure is repeated $256$ times. For the discrete case, we set $E=32$ and use $30$ training matrices per epoch, from which a subset of $S$ entries is randomly selected for training; the parameters of each thermal state are sampled from $T = [0.1, 5]$, $J = [0.1, 2]$, $H = [0.1, 2]$, and $g_k = \begin{cases} \text{Uniform}(-0.3, 0.3), & \text{with probability } 0.5, \\ 0, & \text{otherwise} \end{cases}$. Consequently, the total number of training samples per epoch is $30 \times S$. During testing, $100$ matrices are randomly generated, with $N$ elements in each matrix used as in-context samples and the remaining $S - N$ elements serving as queries. Model performance is evaluated using the normalized mean squared error (NMSE), defined as the squared prediction error of the query output in \Cref{thm:generalization_transformer_nqs}, normalized by the squared magnitude of the corresponding ground-truth output, and averaged over all test queries.
In addition, although the theoretical analysis is carried out using the augmented embedding matrix in \eqref{the construction of embedding vector each layer} of Appendix~\ref{Proof of transformer emulate proximal GD}, the simulations are performed using the practical embedding matrix in \eqref{the construction of embedding vector}. The experimental results confirm that the predicted scaling behaviors are preserved under the practical embedding construction. All models are initialized using Kaiming initialization \cite{he2015delving}, and all results are averaged over $20$ independent Monte Carlo trials.

\begin{figure}[t]
\centering
\subfigure[]{
\begin{minipage}[t]{0.31\textwidth}
\centering
\includegraphics[width=5.5cm]{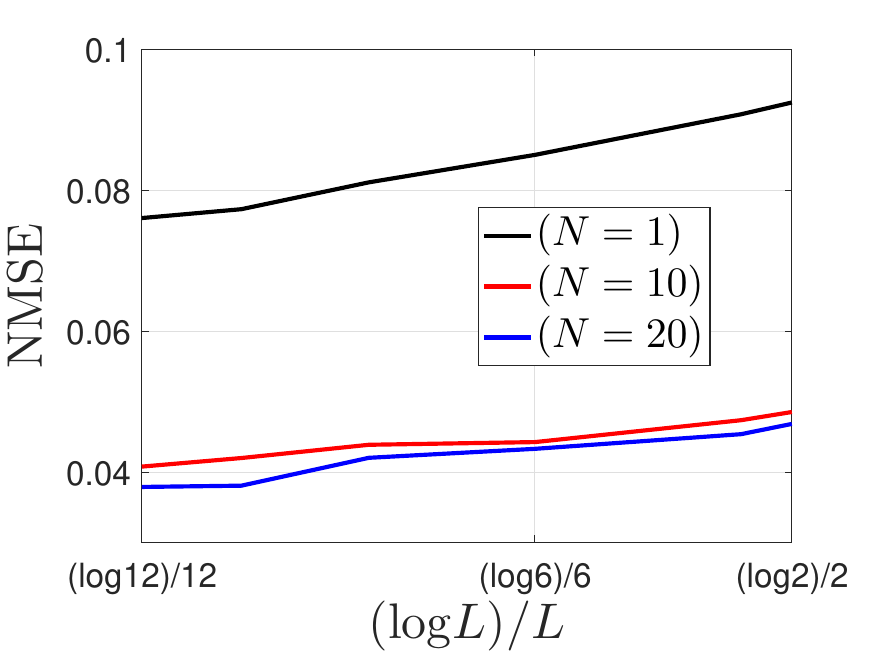}
\end{minipage}
\label{Gaussian diff L new}
}
\subfigure[]{
\begin{minipage}[t]{0.31\textwidth}
\centering
\includegraphics[width=5.5cm]{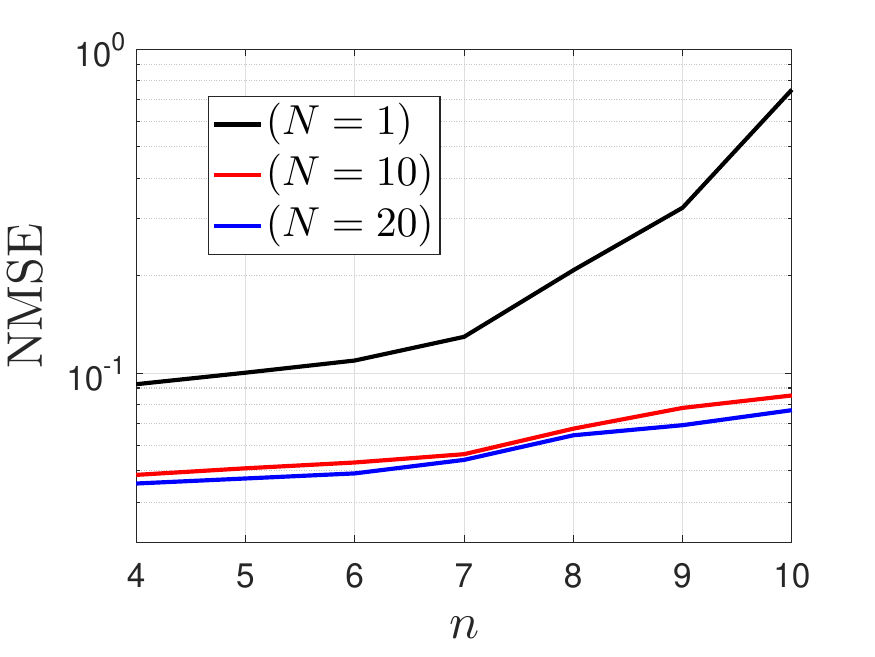}
\end{minipage}
\label{Gaussian diff n new}
}
\subfigure[]{
\begin{minipage}[t]{0.31\textwidth}
\centering
\includegraphics[width=5.5cm]{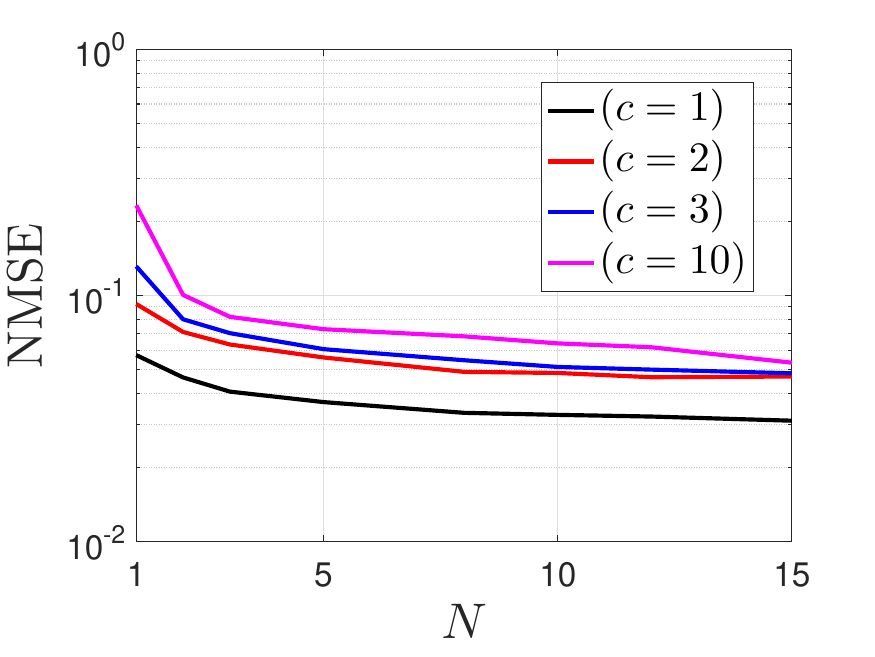}
\end{minipage}
\label{Gaussian diff c new}
}
\caption{Comparison of generalization errors of Transformer-based neural quantum states in the continuous case over 500 training epochs with Transformer width 64: (a) varying $L$ and $N$ with $n=4$ and $c=2$, (b) varying $n$ and $N$ with $c=2$ and $L=2$, (c) varying $c$ and $N$ with $n=4$ and $L=2$. These results empirically validate the theoretical scaling law in \eqref{Upper bound of the estimated output and ground truth output main paper original}: the generalization error decreases with the number of in-context samples $N$ and increases with the problem dimension $n$ and energy constraint $B$, while exhibiting an approximately linear dependence on $(\log L)/L$.}
\end{figure}

In the first experiment, we evaluate the generalization error of neural quantum states in the continuous setting. As shown in \Cref{Gaussian diff L new}, the error increases as $(\log L)/L$ increases or $N$ decreases, consistent with the scaling behavior in \eqref{Upper bound of the estimated output and ground truth output main paper original} discussed after \Cref{thm:generalization_transformer_nqs}. We note that the empirical dependence on $(\log L)/L$ exhibits an approximately linear trend, in agreement with the theoretical scaling law in \eqref{Upper bound of the estimated output and ground truth output main paper original}. This trend reflects the interplay between model capacity and available information: deeper Transformers ($L$ larger) provide stronger representational power, while more in-context samples ($N$ larger) supply richer information predicting components of the quantum state not directly specified in the in-context samples. Empirically, increasing $N$ leads to a more pronounced reduction in error than increasing $L$, indicating that data availability is the dominant factor in this regime. Furthermore, \Cref{Gaussian diff n new} shows that the error increases with $n$. This is because a larger $n$ expands the space of admissible quantum wavefunctions, thereby increasing the intrinsic difficulty of the prediction problem and leading to higher generalization error. We further observe an approximately linear trend for $N=10$ and $N=20$, whereas this linear relationship breaks down for $N=1$, due to the insufficient number of in-context samples, which is insufficient to reliably capture the underlying mapping. Finally, we compare the generalization error under different values of $c$ and $N$. The results indicate that the generalization error increases as $c$ increases (equivalently, as $B$ increases), or as $N$ decreases. In particular, the dependence on $c$ can be interpreted as a consequence of the enlarged support of the input distribution: a larger $c$ induces a broader class of admissible inputs, which in turn increases the complexity of the target function class and yields a more difficult generalization problem under a fixed training budget.

In the second experiment, we evaluate the generalization error of neural quantum states in the discrete setting. As shown in \Cref{Thermal diff L diff n}, consistent with the continuous case, the generalization error decreases as $L$ and $N$ increase or $n$ decreases, indicating that larger models and lower system complexity improve generalization. To ensure a fair comparison, we use 400 training epochs for $n=5$, compared to 300 epochs for $n=4$, as the former requires more optimization steps to reach a comparable training regime. Furthermore, \Cref{Thermal diff L diff sample} investigates the effect of the number of sampled density matrix elements, since training on the full density matrix is computationally infeasible. For a shallow Transformer ($L=1$), using $S=128$ samples achieves performance comparable to $S=512$, suggesting that a small sample size is sufficient. In contrast, for a deeper Transformer ($L=5$), performance improves significantly with $S=512$, indicating that larger models require more samples to fully exploit their representational capacity.

\begin{figure}[t]
\centering
\subfigure[]{
\begin{minipage}[t]{0.45\textwidth}
\centering
\includegraphics[width=7cm]{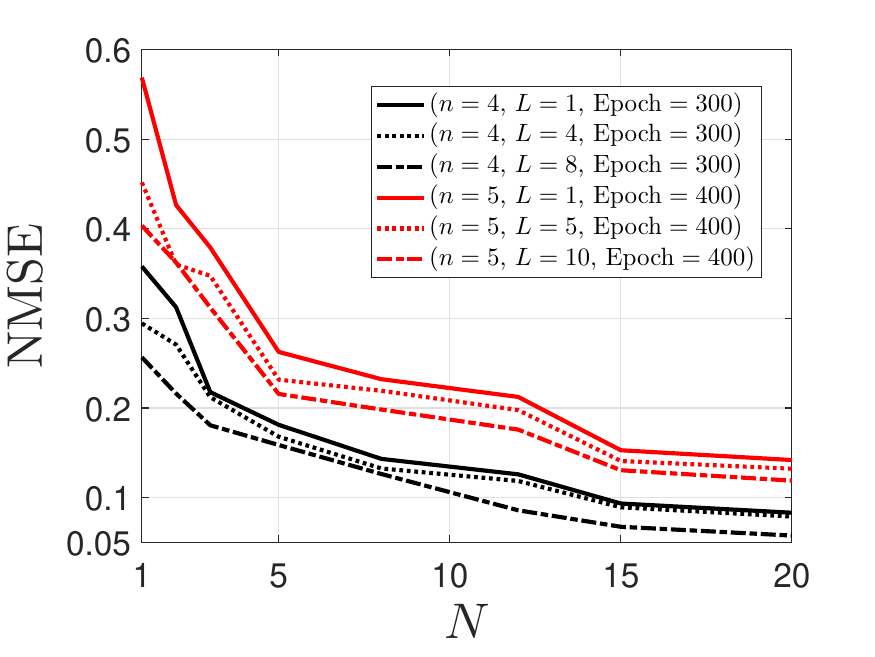}
\end{minipage}
\label{Thermal diff L diff n}
}
\subfigure[]{
\begin{minipage}[t]{0.45\textwidth}
\centering
\includegraphics[width=7cm]{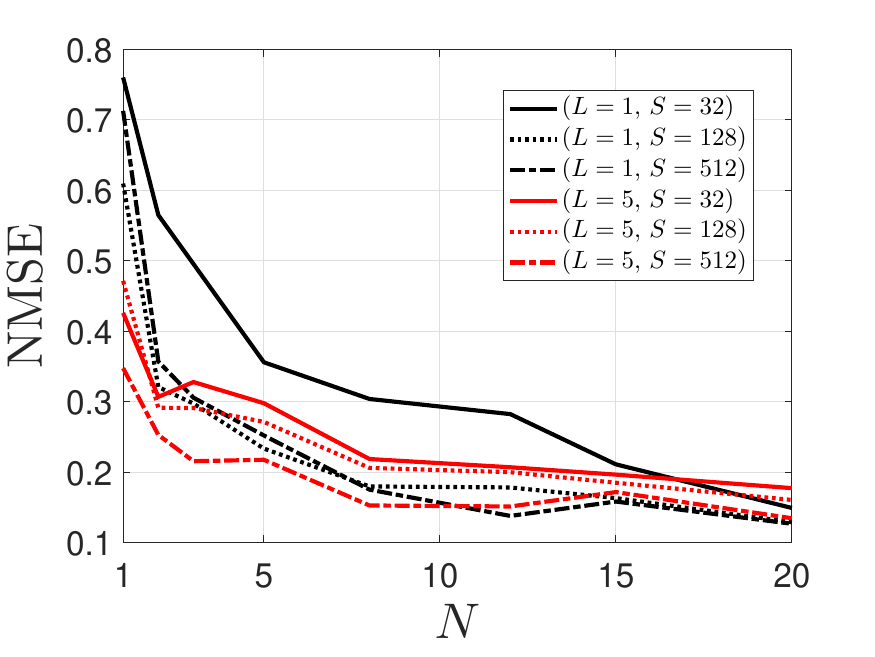}
\end{minipage}
\label{Thermal diff L diff sample}
}
\caption{Comparison of generalization errors of Transformer-based neural quantum states in the discrete case with Transformer width 32: (a) varying $n$ and $L$ with $S=128$, (b) varying $L$ and $S$ with $n=5$ and 500 training epochs. These results demonstrate that the generalization error decreases with increasing model depth and the number of in-context samples, while increasing with system size $n$, consistent with the scaling behavior predicted by the theoretical analysis.}
\end{figure}

In the third experiment, we validate the conclusion established in \eqref{Upper bound of the estimated output and ground truth output main paper quantum property}. Specifically, we consider the expectation value of an a priori unknown linear observable \cite{huang2020predicting}. For an $n$-qubit thermal state $\vrho^\star$ drawn from the disordered transverse-field Ising Hamiltonian family introduced above, the target property is defined as $f_{\vtheta}(\vrho^\star)=\trace(\vrho^\star \mO_{\vtheta})$ where $\mO_{\vtheta}=\sum_{k=1}^{D_\theta}\theta_k\mP_k$. Here $\{\mP_k\}_{k=1}^{D_\theta}$ is a fixed basis consisting of local and two-body Pauli operators, including the $n$ single-site operators $\{\mZ_k\}_{k=1}^n$, the $n$ single-site operators $\{\mX_k\}_{k=1}^n$, and the $n-1$ nearest-neighbor correlators $\{\mZ_k\mZ_{k+1}\}_{k=1}^{n-1}$. Consequently, the basis contains a total of $D_{\theta}=3n-1$ operators. The coefficient vector $\vtheta = \begin{bmatrix} \theta_1 & \cdots & \theta_{D_\theta} \end{bmatrix}^\top\in\R^{D_\theta}$ is sampled uniformly from $[-1,1]^{D_\theta}$, and the resulting observable $\mO_{\vtheta}$ is fixed across all $N$ in-context examples and the query within each task. The thermal states are generated using the same procedure as in the previous experiments, except that the disorder parameters $g_k$ are sampled independently and uniformly from $[-1,1]$. To reduce the input domain, each in-context example consists of the MPO representation of the Hamiltonian $\mH_i$, together with its corresponding temperature $T_i$, as the input, and the associated property value $y_i = f_{\vtheta}(\vrho_i^\star)$ as the output. The query is represented by the same input variables, namely its MPO tensors and temperature, but without the corresponding property value. The model is trained to predict $y_{N+1}$ from the $N$ labeled in-context examples and the unlabeled query. The prediction error is measured by the NMSE, defined as the squared difference between the predicted and ground-truth linear observable values, normalized by the squared ground-truth value. As illustrated in \Cref{fig:comparison between quantum property with diff n}, the NMSE decreases monotonically as the number of in-context samples $N$ increases, consistently across all system sizes $n$. This indicates that the Transformer effectively benefits from additional in-context information when predicting unknown linear observables. Moreover, for a fixed model capacity (Transformer width 64 and depth $L=4$) and sufficiently large $N$, larger systems exhibit higher NMSE, suggesting that the task becomes intrinsically more challenging as the system dimension increases under the same model and data budget. Both trends are consistent with the theoretical bounds established in \eqref{Upper bound of the layers in the Transformer normalized quantum property} and \eqref{Upper bound of the estimated output and ground truth output main paper quantum property}.

\begin{figure}[!ht]
  \centering
  \includegraphics[width=6.5cm]{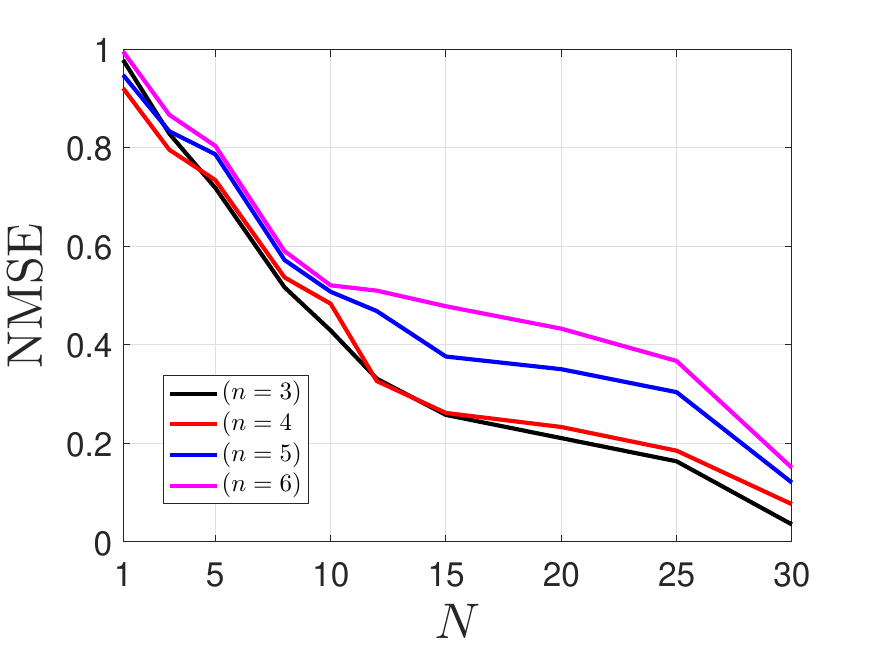}
  \caption{Comparison of the generalization error for predicting unknown linear observables as the number of qubits $n$ varies. The Transformer width is fixed at 64 and the number of layers is fixed at $L=4$. Results are obtained after 300 training epochs (32 batches per epoch) and evaluated on 100 testing tasks. These results demonstrate the predicted benefit of increasing the number of in-context samples $N$ and the increasing difficulty of the prediction task with system size $n$, consistent with the theoretical bounds in \eqref{Upper bound of the layers in the Transformer normalized quantum property} and \eqref{Upper bound of the estimated output and ground truth output main paper quantum property}.}
  \label{fig:comparison between quantum property with diff n}
\end{figure}

\section{Conclusion}

In this paper, we theoretically investigate the generalization behavior of Transformer-based neural quantum states under ICL. By establishing a rigorous inference-time generalization error bound, we show that the prediction error decreases inversely with both the number of in-context examples and the depth of the Transformer architecture. Our analysis further demonstrates that achieving such guarantees requires a Transformer depth that scales linearly with the system size, measured either by the number of particles in continuous systems or the number of qudits in discrete systems. Furthermore, our results provide a quantitative characterization of how in-context information can be exploited at inference time to improve prediction accuracy, offering a principled perspective on data-dependent generalization in neural quantum state representations. Taken together, these findings contribute to a theoretical understanding of ICL in quantum many-body modeling and highlight the scalability of Transformer-based architectures in high-dimensional quantum systems. Numerical experiments further support our theoretical predictions and confirm the derived scaling laws.

A number of theoretical questions remain open. The present analysis is established through an augmented embedding construction that facilitates the proof, whereas the numerical experiments are conducted using the practical embedding matrix. Although the empirical results exhibit the same scaling behavior under the practical embedding, extending the current theoretical framework to rigorously analyze this practically relevant setting remains an important direction for future research. Moreover, an important direction is to identify assumptions that are intrinsic to quantum learning problems, thereby replacing the current proof-oriented assumptions with conditions that more precisely capture the underlying quantum structure. Finally, while the present work establishes the existence of Transformer architectures achieving the desired inference-time generalization guarantees, it remains an important open question whether practical training algorithms can be shown to converge to such solutions with provable guarantees.

\section{Acknowledgments}
\label{sec: ack}

We acknowledge funding support from DARPA HR0011-25-2-0042, ARO W911NF-23-1-0343, NSF ECCS-2532643 and CCF-2246213.  ZQ gratefully acknowledges support from the MICDE Research Scholars Program at the University of Michigan. Finally, we are grateful to Gen Li and Yuchen Jiao for valuable discussions regarding the proof.

\newpage

\appendices

\section{Proof of Theorem~\ref{thm:generalization_transformer_nqs}}
\label{sec: proof of main theorem}

Before proving \Cref{thm:generalization_transformer_nqs}, we first restate the theorem below.
\begin{theorem}
\label{Lemma of final conclusion of all lemmas}
Suppose that the in-context samples and the query input are drawn i.i.d. from an underlying distribution supported on either the continuous domain $\{ (\vx_1,\dots,\vx_n): \sum_{i=1}^n \|\vx_i\|_2^2 \le B, \vx_i\in\R^{D}, i\in[n] \}$ or the discrete domain  $\{ (x_1,\dots,x_n): x_i \in \{0,\dots,d-1\},\ i \in [n]\}$ and that each component of $\vd(\cdot)$ belongs to a function class $\mathcal{G}$ whose elements have bounded first absolute moment on the corresponding domain. Consider a Transformer model whose number of attention heads in the $\ell$-th layer is given by
\begin{eqnarray}
    \label{requirement of attenstion layers appendix_1}
    M_\ell=
\begin{cases}
1, & \ell=1,\\
4, & \ell>1,\ \ell\ \text{is odd},\\
8, & \ell\ \text{is even},
\end{cases}
\end{eqnarray}
with input dimension $nD+2P_1+2P_2+12$.
Assume further that
\begin{eqnarray}
    \label{Upper bound of the layers in the Transformer appendix}
    \begin{cases}
    L\geq B, \ \frac{L}{\log L} \geq \frac{2nDB\sqrt{N}}{\sqrt{\log N}} & \textup{continuous domain},   \\
    L\ge d, \ \ \frac{L}{\log L} \geq \frac{2nd\sqrt{N}}{\sqrt{\log N}}, & \textup{discrete domain},
    \end{cases}
\end{eqnarray}
and that
\begin{eqnarray}
    \label{requirement Pk of the layers in the Transformer}
    P_k = \begin{cases}
    \Omega(nDB\max\{\log(BL), \log(B\sqrt{N}/\sqrt{\log N}) \}),   \\
    \Omega(nd\max\{\log(dL), \log(d\sqrt{N}/\sqrt{\log N}) \} ),
    \end{cases}
\end{eqnarray}
where the first and second cases correspond to the continuous and discrete domains, respectively. Suppose that the in-context samples and the query input are drawn i.i.d. from an underlying distribution that shares the same structural properties as the training distribution, forming an in-context sequence of size $N$ together with a query input. Denote the transformer prediction on the query by $\widehat{\vy}_{N+1} \in \R^2$. Then, with high probability, the expected prediction error satisfies
\begin{eqnarray}
    \label{Upper bound of the estimated output and ground truth output main paper}
    \begin{cases}
    \E[\|\wh \vy_{N+1} -  \vd(\{\vx_{N+1,j}\}_{j=1}^{n} )\|_2^2]= O\big(\sqrt{\frac{\log N}{N}} \big), \\
    \E[\|\wh \vy_{N+1} -  \vd(\{x_{N+1,j}\}_{j=1}^{n} )\|_2^2]= O\big( \sqrt{\frac{\log N}{N}}  \big),
    \end{cases}
\end{eqnarray}
for the continuous and discrete domains respectively.
\end{theorem}

Following the proof strategy of \cite{li2025transformers}, we divide the proof into four main parts. First, we construct a universal feature representation for the target function class and show that the target function admits a sparse linear representation under the resulting features. Second, we establish that the corresponding sparse coefficients can be recovered by the Lasso estimator. Third, we prove that the inexact proximal gradient method converges to the Lasso solution. Fourth, we construct a Transformer that emulates the inexact proximal gradient iterations, thereby establishing the desired generalization guarantee. Finally, by combining these approximation results, we derive the desired conclusion. The results established in the second and third steps directly follow from existing theoretical analyses in \cite{li2025transformers}. Nevertheless, for completeness and clarity, we briefly summarize the key arguments underlying these steps. In the analysis below, we primarily focus on continuous degrees of freedom. The discrete system can be viewed as a special case of the continuous one via discretization; therefore, the theoretical results developed for continuous quantum wavefunctions naturally extend to the discrete setting. Without loss of generality, we assume that there exists a finite constant $B$ such that, $(\vx_1,\vx_2,\dots,\vx_n)\in \calB = \{\vx_i\in\R^{D}, i\in[n]: \sum_{i=1}^n\|\vx_i\|_2^2\leq B\}$, which is physically reasonable since the wavefunction typically decays rapidly outside a finite region (e.g., due to finite energy) and is also convenient for establishing boundedness in the subsequent theoretical analysis. For the final step, we explicitly construct the Transformer to emulate the inexact proximal gradient iterations. In particular, we show in {Appendix}~\ref{Proof of transformer emulate proximal GD} that the layer-wise attention-head configuration specified in~\eqref{requirement of attenstion layers appendix} is sufficient for this construction.

Throughout the proof, we directly invoke previously established results whenever they are applicable. The arguments that require new technical developments or differ from existing analyses are explicitly highlighted, and complete proofs are provided for these components.

\subsection{Universal Feature Representation of the Target Function}
\label{sec: universal feature representation}

In this subsection, we construct a universal feature representation for the target function class under the assumed bounded domain. Specifically, we show that each target function can be expressed as a linear combination of the constructed features with bounded coefficients, which provides the foundation for the subsequent Lasso-based sparse approximation. To simplify the notation, we define
\begin{eqnarray}
\label{definition of vector z}
\vz = \begin{bmatrix} \vx_{1}^\top & \cdots & \vx_{n}^\top \end{bmatrix}^\top \in \R^{nD},
\end{eqnarray}
and introduce the functions $g_1(\vz) = d_1(\vx_{1}, \dots, \vx_{n})$ and $g_2(\vz) = d_2(\vx_{1}, \dots, \vx_{n})$. Next, let $\calG$ denote the class of functions with bounded first absolute moment on $\{ \vz: \|\vz\|_2^2 \le B \}$, and assume that $g_1, g_2 : \R^{nD} \to \R$ belong to $\calG$. This assumption ensures that their Fourier representations exist.  For $k=1,2$, let $F_{g_k}$ denote the Fourier transform of $g_k$. Then $g_k$ admits the following Fourier representation:
\begin{eqnarray}
    \label{Fourier transform}
    F_{g_k}(\boldsymbol{\omega}) &\!\!\!\!=\!\!\!\!& \frac{1}{2\pi} \int_{\vz} e^{-i \boldsymbol{\omega}^\top \vz} {g_k}(\vz)\, \textup{d}{\vz}, \\
    \label{inverse Fourier transform}
    {g_k}(\vz) &\!\!\!\!=\!\!\!\!& \int_{\boldsymbol{\omega}} e^{i \boldsymbol{\omega}^\top \vz} F_{g_k}(\boldsymbol{\omega})\, \textup{d}\boldsymbol{\omega},
\end{eqnarray}
where $i = \sqrt{-1}$ denotes the imaginary unit. Following the framework in \cite{barron1993universal,li2025transformers}, we introduce the following complexity measure:
\begin{eqnarray}
    \label{C1 definition}
    C_{g_k} = \sup_{g_k\in\calG} |g_k({\bm 0}) | + \int_{\boldsymbol{\omega}}\sup_{g_k\in\calG}\|F_{\nabla g_k}(\boldsymbol{\omega})  \|_2\textup{d}\boldsymbol{\omega}< \infty, \ k=1,2,
\end{eqnarray}
which controls the first moment of the Fourier magnitude distribution over the function class $\calG$ and characterizes the $\ell_1$ norm of the maximal Fourier magnitude of the gradient. Recall that the input space has been restricted to the Euclidean ball $\calB = \{\vz = \begin{bmatrix} \vx_{1}^\top   & \cdots & \vx_{n}^\top \end{bmatrix}^\top \in\R^{nD}: \|\vz\|_2^2 = \sum_{i=1}^n\|\vx_i\|_2^2\leq B\}$. We next introduce a covering structure for the joint space $\calG \times \calB$. A finite set $\calN_{\epsilon_k}$ is said to be an $\epsilon_k$-cover of $\calG \times \calB$ if, for every $(g_k,\vz)\in\calG\times\calB$, there exists $(\wh g_k,\wh{\vz})\in\calN_{\epsilon_k}$ such that
\begin{eqnarray}
    \label{covering inequality}
    \|\vz - \wh {\vz}\|_2\leq \epsilon_k, \ \ \ |g_k(\vz) - g_k({\bm 0}) -  \wh g_k(\wh {\vz}) + \wh g_k({\bm 0})|\leq C_{g_k} \epsilon_k, \  k=1,2.
\end{eqnarray}
Before stating the main result, following \cite[eq. (19)]{li2025transformers}, we introduce the sigmoid function as a smooth approximation of the step function. This function is used to construct a finite set of features for approximating $g_k(\vz) - g_k({\bm 0})$. Specifically, let
\begin{eqnarray}
    \label{smoothed_step_function}
    \phi(z)
= \left( z + \frac{1}{2} \right) \mathbbm{1}\!\left\{ z + \frac{1}{2} > 0 \right\}
 - \left( z - \frac{1}{2} \right) \mathbbm{1}\!\left\{ z - \frac{1}{2} > 0 \right\},
\end{eqnarray}
which satisfies $\lim_{z \to -\infty} \phi(z) = 0$, $\lim_{z \to \infty} \phi(z) = 1$, and $\phi(0) = \frac{1}{2}$.  Based on this construction, we will define a finite collection of features that provides an approximate representation of $g_k(\vz) - g_k({\bm 0})$.
\begin{lemma}
\label{feature_construction_of_functions}
Let $\tau_k > 2$ and $P_k \ge c_k \log |\calN_{\epsilon_k}|$ for some constants $c_k > 0$, $k=1,2$. Then there exists a finite collection of functions $\{\phi_{k,j}: \R^{nD} \to \R \}_{1 \le j \le P_k},  k=1,2$ such that for any $g_k \in \calG$ and any $\vz \in \calB$, with high probability,
\begin{eqnarray}
    \label{difference_of_functions_and_feature}
    \bigg| g_k(\vz) - g_k({\bm 0}) - \frac{1}{P_k} \sum_{j=1}^{P_k} a_{k,j}^\star \phi_{k,j}(\vz) \bigg|
    \lesssim C_{g_k} \left( \frac{\sqrt{B}}{\tau_k} + B^{\frac{1}{3}} \left(\frac{\log |\calN_{\epsilon_k}|}{P_k}\right)^{\frac{1}{3}} + \tau_k \epsilon_k \right),
\end{eqnarray}
for some coefficients $\{a_{k,j}^\star\}_{j=1}^{P_k}$ satisfying $ |g_k({\bm 0})| + \frac{1}{P_k} \sum_{j=1}^{P_k} |a_{k,j}^\star| \lesssim 4 C_{g_k} $. The feature functions are constructed as $\phi_{k,j}(\vz) = \phi ( \tau_k ( \frac{\vw_{k,j}^\top \vz}{\|\vw_{k,j}\|_2} - t_{k,j} ))$ for some constants $\{(\vw_{k,j}, t_{k,j})\}_{1 \le j \le P_k} \subset \R^{nD} \times \R$, $k=1,2$.
\end{lemma}
The proof is provided in Appendix~\ref{Proof of feature_construction_of_functions}. This lemma establishes that, for any $g_k \in \calG$ and input $\vz \in \calB$, the deviation $g_k(\vz) - g_k({\bm 0})$ can be effectively captured by a linear combination of the constructed feature functions $\{\phi_{k,j}(\vz)\}$. For the subsequent analysis, we select the parameters $\tau_k = \frac{B^{\frac{1}{4}}}{\sqrt{\epsilon_k}}$ which induces a corresponding constraint on $\epsilon_k$ and consequently affects the choice of $P_k$ through its dependence on $B$; a detailed discussion is deferred to Appendix~\ref{sec: Composition of Approximations}.  Our analysis extends the result in \cite[Lemma 1]{li2025transformers} from the unit ball to a general bounded domain with radius $B$.  When $B=1$, our result recovers the corresponding result in \cite[Lemma 1]{li2025transformers}.

We further choose $\eta_k = C_{g_k} ( 2 B^{\frac{1}{4}} \sqrt{\epsilon_k} + B^{\frac{1}{3}} (\frac{\log |\calN_{\epsilon_k}|}{P_k})^{\frac{1}{3}} )$, which allows us to simplify the bound \eqref{difference_of_functions_and_feature} as $|g_k(\vz) - g_k({\bm 0}) - \frac{1}{P_k}\sum_{j=1}^{P_k} a_{k,j}^\star \phi_{k,j}(\vz)   | \lesssim  \eta_k$. Hence, Lemma~\ref{feature_construction_of_functions} can be specialized to yield the following explicit form:
\begin{lemma}
\label{feature_construction_special_case}
Under the setting of Lemma~\ref{feature_construction_of_functions}, let the parameters be chosen as
\begin{eqnarray}
    \label{choices of parameters}
    \tau_k = \frac{B^{\frac{1}{4}}}{\sqrt{\epsilon_k}}, \qquad
    \eta_k = C_{g_k} \bigg( 2 B^{\frac{1}{4}} \sqrt{\epsilon_k} + B^{\frac{1}{3}} \bigg(\frac{\log |\calN_{\epsilon_k}|}{P_k}\bigg)^{\frac{1}{3}} \bigg),
\end{eqnarray}
Then, for any $g_k \in \calG$ and $\vz \in \calB$, there exists a linear combination of the feature functions $\{\phi_{k,j}(\vz)\}_{j=1}^{P_k}$ such that
\begin{eqnarray}
    \label{final conclusion of Lemma 1}
    \bigg|g_k(\vz) - g_k({\bm 0}) - \frac{1}{P_k}\sum_{j=1}^{P_k} a_{k,j}^\star \phi_{k,j}(\vz)   \bigg| \lesssim  \eta_k.
\end{eqnarray}
with coefficients $\{a_{k,j}^\star\}$ satisfying $ |g_k({\bm 0})| + \frac{1}{P_k} \sum_{j=1}^{P_k} |a_{k,j}^\star| \lesssim 4 C_{g_k} $.
\end{lemma}

\subsection{Learning Sparse Coefficients via Lasso}
\label{sec:Approximation of the Target Function Class by the Lasso Estimator}

Building on Lemma~\ref{feature_construction_special_case}, we observe that the target function admits an approximate linear representation in terms of the constructed features with a controlled $\ell_1$-norm on the coefficients. Specifically, the coefficients $\{a_{k,j}^\star\}$ lie within an $\ell_1$-ball of radius $O(C_{g_k})$, as established in Lemma~\ref{feature_construction_special_case}. Consequently, the recovery problem can be formulated as a constrained least-squares estimator over an $\ell_1$-ball, which is equivalent to the standard Lasso formulation. The subsequent step is to employ the Lasso estimator to recover the coefficient vectors $\{a_{k,j}^\star\}$.   For the purpose of the ensuing analysis, we introduce the following vectorized notation: $\va_{k}^\star  = \begin{bmatrix}  a_{k,1}^\star/P_k & \cdots &   a_{k,P_k}^\star/P_k &   g_k({\bm 0})   \end{bmatrix}^\top$ and $\vphi_{k}^{i}  =  \begin{bmatrix} \phi_{k,1}(\vz_i)  & \cdots & \phi_{k,P_k}(\vz_i) & 1    \end{bmatrix}^\top$. We then estimate  the coefficient vectors $\va_k^\star$ from the dataset $\{ (\vz_i, g_k(\vz_i)) \}_{i=1}^{N}$ for $k=1,2$ via the Lasso estimator \cite{tibshirani1996regression}
\begin{eqnarray}
    \label{The Lasso estimator}
    \min_{\va_{k}\in\R^{P_k+1}} h(\va_{k}) \coloneqq \frac{1}{N}\sum_{i=1}^{N}(g_k(\vz_i) - {\vphi_{k}^{i}}^\top \va_{k})^2 + \lambda_k\|\va_{k} \|_1, k=1,2,
\end{eqnarray}
where $\lambda_k, k=1,2$ denote the regularized parameters. The rationale for employing Lasso stems from \Cref{feature_construction_of_functions}, which guarantees that $\|\va_{k}^\star\|_1 = |g_k({\bm 0})| + \frac{1}{P_k}\sum_{j=1}^{P_k} |a_{k,j}^\star|\lesssim  4 C_{g_k}$. Accordingly, we require $\lambda_k \gtrsim \eta_k^2 / C_{g_k}$, which ensures a proper balance between the empirical risk and the $\ell_1$ regularization, consistent with the intrinsic scale of the oracle coefficient vector and the feature approximation error $\eta_k$.  In addition, \eqref{final conclusion of Lemma 1} establishes that $\va_k^\star$ yields an accurate approximation of the target function values $\{ g_k(\vz_i) \}$ on the dataset. Consequently, the subsequent analysis focuses on bounding the prediction error of any estimator $\hat{\va}_k$ that satisfies $\wh \va_{k}$ obeying $h(\wh\va_{k}) - h(\va_{k}^\star) \leq \epsilon_{k,\rm{opt}}, k=1,2$, thereby quantifying the performance of near-optimal Lasso solutions relative to the true coefficient vectors.
\begin{lemma}
\label{Generalization error bound of Lasso estimator}
Let $\lambda_k \ge \Omega(\eta_k^2 / C_{g_k})$, where $\eta_k$ is defined in \eqref{choices of parameters}. Consider any estimator $\wh\va_{k}$ that is independent of $\vz_{N+1}$ and satisfies $h(\wh\va_{k}) - h(\va_k^\star) \le \epsilon_{k,\rm opt}$ for $k=1,2$. Then, with high probability, its $\ell_1$ norm is controlled as $\|\wh\va_{k}\|_1 \lesssim C_{g_k} + \lambda_k^{-1} \epsilon_{k,\rm opt}$, and the expected squared prediction error on a new input $\vz_{N+1}$ satisfies
\begin{eqnarray}
\label{Generalization error of input at N plus 1 1}
\E[({\vphi_{k}^{N+1}}^\top \wh\va_{k} - g_k(\vz_{N+1}))^2 ] \lesssim  \sqrt{\frac{\log N}{N}}(C_{g_k}^2 + \lambda_k^{-2}\epsilon_{k,\rm{opt}}^2 ) + \eta_k^2 + \lambda_k C_{g_k} + \epsilon_{k,\rm{opt}}.
\end{eqnarray}
Here, the expectation is taken over the randomness of $\vz_{N+1}$.
\end{lemma}
The proof of \Cref{Generalization error bound of Lasso estimator} follows directly from the noiseless setting considered in \cite[Lemma 2]{li2025transformers}. Therefore, we omit the proof for brevity. Note that the parameters in \eqref{Generalization error of input at N plus 1 1} will be specified in {Appendix}~\ref{sec: Composition of Approximations}, which then determines the prediction error of the Lasso estimator as stated in \eqref{upper bound of first term in the final conclusion}.

\subsection{Approximating the Lasso Solution via Inexact Proximal Gradient Method}
\label{sec: convergence of IPGM}

In this subsection, we show that the Lasso solution can be approximated by a finite number of inexact proximal gradient iterations. Although the prediction error of the Lasso estimator has been characterized in the previous subsection, the optimization problem in \eqref{The Lasso estimator} generally does not admit a closed-form solution. To address this issue, we consider an approximate solution obtained via the (inexact) proximal gradient method (PGM). Specifically, starting from the initialization $\va_{k}^{(0)} = {\bf 0}$, the iterative update at step $t = 1,\dots,T$ takes the form
\begin{eqnarray}
    \label{The update of PGM}
    \va_{k}^{(t)} = \calP_{\mu_k\lambda_k}\bigg(\va_{k}^{(t-1)} - \frac{2\mu_k}{N}\sum_{i=1}^{N}({\vphi_{k}^{i}}^\top \va_{k}^{(t-1)} - g_k(\vz_i) )\vphi_{k}^{i}  \bigg) + \ve_k^{(t)},
\end{eqnarray}
where $\mu_k$ is the step size, $\calP_{\mu_k\lambda_k}(z) \coloneqq \sign(z)\max\{|z| - \mu_k\lambda_k, 0  \}$ represents the entrywise soft-thresholding operator, and $\ve_k^{(t)}$ captures the inexactness of the update. It will be shown that after $T = \frac{L-1}{2}$ iterations, the resulting iterate $\va_k^{(T)}$ satisfies the following guarantee:
\begin{lemma}
\label{Performance analysis of PGM}
Suppose the proximal gradient method (PGM) is initialized at $\va_k^{(0)} = {\bf 0}$ and run for
$T = \frac{L-1}{2}$ iterations with step size $\mu_k = \frac{1}{2P_k}$.
Assume the regularization parameter satisfies $\lambda_k \gtrsim \sqrt{\frac{\log N}{N}}\,(C_{g_k}+\eta_k)$, and the inexact updates obey $\|\ve_k^{(t)}\|_1 \le \epsilon_{k,\rm app}$ for all $t\le T$, where $\epsilon_{k,\rm app} \lesssim C_{g_k}$. Then, with high probability, the iterate $\va_k^{(T)}$ satisfies
\begin{eqnarray}
    \label{The update of PGM final step}
    h(\va_{k}^{(T)}) - h(\va_{k}^\star)  \leq \frac{c_1P_kC_{g_k}^2}{L} + c_1(L + P_k)\epsilon_{k,\rm{app}} (C_{g_k} + \eta_k+  \max_{1\leq j\leq T}\|\va_{k}^{(j)} \|_1  + \lambda_k),
\end{eqnarray}
where $c_1$ is a universal constant. Moreover, the intermediate iterates satisfy the following bound for all $t\le T$: $\|\va_{k}^{(t)}\|_1 \lesssim C_{g_k} + \frac{P_k C_{g_k}^2}{\lambda_k t} + \sqrt{N}(P_k + t)\max_{1\leq j\leq t} \|\ve_k^{(j)}\|_1 + \frac{P_k + t}{\lambda_k}\max_{1\leq j\leq t} (\|\ve_k^{(j)}\|_1 \|\va_{k}^{(j)} \|_1)$.
\end{lemma}
The proof of \Cref{Performance analysis of PGM} is identical to the corresponding analysis in \cite[Lemma 3]{li2025transformers}. Therefore, we omit the proof here. Note that \Cref{Performance analysis of PGM} yields the prediction error bound for the Lasso estimator obtained via the PGM. After $T=(L-1)/2$ iterations, the iterate $\va_k^{(T)}$ attains the required optimization precision, which establishes the necessary foundation for the ensuing analysis of its Transformer-based approximation.

\subsection{Transformer Implementation of the Inexact Proximal Gradient Method}
\label{sec:Approximation of the Lasso Solution by a Transformer}

In the previous step, we established an optimization error bound for the Lasso estimator obtained via the PGM. We now construct a Transformer architecture that approximates the PGM iterations, thereby establishing a connection between the Transformer and the optimization procedure. Unlike the single-output construction in \cite[Lemma 4]{li2025transformers}, our analysis considers two coupled output functions. Consequently, the input embedding matrix needs to be augmented to simultaneously encode the information associated with both outputs throughout the inexact proximal gradient iterations.  To facilitate the subsequent analysis, we introduce the following input embedding matrix:
\begin{eqnarray}
    \label{the construction of embedding vector each layer initialization}
    \mH^{(0)} := \begin{bmatrix}\vz_{1}  & \cdots & \vz_{N}  & \vz_{N+1}  \\  1  & \cdots & 1  & 1  \\  g_{1}(\vz_1) & \cdots & g_{1}(\vz_N) & 0 \\  g_{2}(\vz_1) & \cdots & g_{2}(\vz_N) & 0 \\  1 & \cdots & 1 & 0 \\ {\bf 0} & \cdots & {\bf 0} & {\bf 0} \end{bmatrix}\in \R^{(nD + 2P_1 + 2P_2 + 12) \times (N+1)}.
\end{eqnarray}
With this construction, the Transformer can emulate the proximal gradient updates, as formalized in the following lemma.
\begin{lemma}
\label{lemma: transformer emulate proximal GD}
Consider a Transformer with $L$ layers, and input embedding matrix $\mH^{(0)}$. Let $M_\ell$ denote the number of
attention heads in the $\ell$-th attention layer, where
\begin{eqnarray}
    \label{requirement of attenstion layers appendix}
    M_\ell=
\begin{cases}
1, & \ell=1,\\
4, & \ell>1,\ \ell\ \text{is odd},\\
8, & \ell\ \text{is even}.
\end{cases}
\end{eqnarray}
The output of the Transformer produces vectors $\vb_1^{(L)}$ and $\vb_2^{(L)}$ that correspond to $\va_1^{(\frac{L-1}{2})}$ and $\va_2^{(\frac{L-1}{2})}$ generated by the (inexact) PGM in \eqref{The update of PGM} for $\lambda_k, k=1,2$, with approximation errors $\ve_k^{(t)}$. Specifically, the parameters can be chosen such that $\max_{1\leq t \leq \frac{L-1}{2}} \| \ve_k^{(t)} \|_1  \lesssim  \frac{C_{g_k}}{(L + P_k)P_k N}, k=1,2$ and
\begin{eqnarray}
    \label{choice of lambda in the transformer}
    \lambda_k &\!\!\!\!\asymp\!\!\!\!& \bigg(\frac{\log N}{N} \bigg)^{\frac{1}{6}}C_{g_k}^{-\frac{1}{3}}\bigg(\sqrt{\frac{\log N}{N}}C_{g_k}^2 + \eta_k^2 + \frac{P_kC_{g_k}^2}{L}    \bigg)^{\frac{2}{3}} + \sqrt{\frac{\log N}{N}}C_{g_k} + C_{g_k}^{-1} \eta_k^2, k = 1,2.
\end{eqnarray}
Under these choices, the following bounds hold:
\begin{eqnarray}
    \label{upper bound of loss at N+1 in the transformer}
    h(\vb_k^{(L)}) - h(\va_k^\star) &\!\!\!\!\lesssim\!\!\!\!& \sqrt{\frac{\log N}{N}}C_{g_k}^2 + \eta_k^2 + \frac{P_k C_{g_k}^2}{L}, k = 1,2,\\
    \label{upper bound of output at N+1 in the transformer}
    ({\vphi_{k}^{N+1}}^\top \vb_k^{(L)} - \wh \vy_{N+1}(k)  )^2  &\!\!\!\!\lesssim\!\!\!\!& \frac{C_{g_k}^2}{N^2}, k=1,2.
\end{eqnarray}
\end{lemma}
The proof is deferred to Appendix~\ref{Proof of transformer emulate proximal GD}. The above result shows that the solution of the Lasso problem can be effectively approximated by a Transformer architecture with depth $L$. Note that, compared with \cite[Lemma 4]{li2025transformers}, our setting requires an extended input embedding matrix to accommodate the two-dimensional Transformer output. This modification changes the intermediate representations, and therefore the corresponding derivations need to be modified accordingly.

\subsection{Composition of Approximations}
\label{sec: Composition of Approximations}

By combining the results established in {Appendices}~\ref{sec: universal feature representation}-\ref{sec:Approximation of the Lasso Solution by a Transformer}, the proof of \Cref{Lemma of final conclusion of all lemmas} follows.
\begin{proof}

According to \eqref{upper bound of loss at N+1 in the transformer} in \Cref{lemma: transformer emulate proximal GD}, we have
\begin{eqnarray}
    \label{upper bound of loss at N+1 in the transformer appendix lemma 5}
    h(\vb_k^{(L)}) - h(\va_k^\star) \leq  \epsilon_{k,\rm{opt}},
\end{eqnarray}
where $\epsilon_{k,\rm{opt}} \asymp \sqrt{\frac{\log N}{N}}C_{g_k}^2 + \eta_k^2 + \frac{P_k C_{g_k}^2}{L}$ and $\vb_k^{(L)} = \wh\va_{k}$. Combining the above bound with
\eqref{Generalization error of input at N plus 1 1} in \Cref{Generalization error bound of Lasso estimator}, we obtain
\begin{eqnarray}
    \label{Generalization error of input at N plus 1 1 appendix lemma 5}
    &\!\!\!\! \!\!\!\!&\E[({\vphi_{k}^{N+1}}^\top \wh\va_{k} - g_k(\vz_{N+1}))^2 ]\nonumber\\
    &\!\!\!\!\lesssim\!\!\!\!& \sqrt{\frac{\log N}{N}}(C_{g_k}^2 + \lambda_k^{-2}\epsilon_{k,\rm{opt}}^2 ) + \eta_k^2 + \lambda_k C_{g_k} + \epsilon_{k,\rm{opt}}\nonumber\\
    &\!\!\!\!\lesssim\!\!\!\!& \sqrt{\frac{\log N}{N}} C_{g_k}^2 + C_{g_k}^2 B^{\frac{1}{2}} \epsilon_k + C_{g_k}^2B^{\frac{2}{3}}\bigg(\frac{\log |\calN_{\epsilon_k}|}{P_k}\bigg)^{\frac{2}{3}} + \lambda_k C_{g_k} + \frac{P_k C_{g_k}^2}{L} + \sqrt{\frac{\log N}{N}}\lambda_k^{-2}\epsilon_{k,\rm{opt}}^2\nonumber\\
    &\!\!\!\!\lesssim\!\!\!\!& \sqrt{\frac{\log N}{N}}C_{g_k}^2   + \frac{P_kC_{g_k}^2}{L} + C_{g_k}^2 B^{\frac{1}{2}} \epsilon_k + C_{g_k}^2B^{\frac{2}{3}}\big(\frac{\log |\calN_{\epsilon_k}|}{P_k}\big)^{\frac{2}{3}},
\end{eqnarray}
where the second inequality follows from \eqref{choices of parameters} together with $\epsilon_{k,\rm{opt}} \asymp \sqrt{\frac{\log N}{N}}C_{g_k}^2 + \eta_k^2 + \frac{P_k C_{g_k}^2}{L}$. For the final inequality, we use the choice of $\lambda_k$ in
\eqref{choice of lambda in the transformer}, which implies $\lambda_k C_{g_k} \lesssim (\frac{\log N}{N} )^{\frac{1}{6}}C_{g_k}^{\frac{2}{3}}(\sqrt{\frac{\log N}{N}}C_{g_k}^2 + \eta_k^2 + \frac{P_kC_{g_k}^2}{L}   )^{\frac{4}{3}} + \sqrt{\frac{\log N}{N}}C_{g_k}^2 +  \eta_k^2 \lesssim (\frac{\log N}{N} )^{\frac{1}{2}}C_{g_k}^2 + \big(\sqrt{\frac{\log N}{N}}C_{g_k}^2 + \eta_k^2 + \frac{P_kC_{g_k}^2}{L}\big) + \sqrt{\frac{\log N}{N}}C_{g_k}^2 +  \eta_k^2 \lesssim \sqrt{\frac{\log N}{N}}C_{g_k}^2   + \frac{P_kC_{g_k}^2}{L} + C_{g_k}^2 B^{\frac{1}{2}} \epsilon_k + C_{g_k}^2B^{\frac{2}{3}}\big(\frac{\log |\calN_{\epsilon_k}|}{P_k}\big)^{\frac{2}{3}}$, where the last step again uses \eqref{choices of parameters}. Moreover, based on \eqref{choice of lambda in the transformer}, $\lambda_k^{-2}\epsilon_{k,\rm{opt}}^2\lesssim  \big(\frac{\log N}{N} \big)^{\frac{1}{6}} C_{g_k}^{\frac{2}{3}} \big(\sqrt{\frac{\log N}{N}}C_{g_k}^2 + \eta_k^2 + \frac{P_k C_{g_k}^2}{L}\big)^{\frac{2}{3}} \lesssim \sqrt{\frac{\log N}{N}}C_{g_k}^2 + \eta_k^2 + \frac{P_k C_{g_k}^2}{L} \lesssim \sqrt{\frac{\log N}{N}}C_{g_k}^2   + \frac{P_kC_{g_k}^2}{L} + C_{g_k}^2 B^{\frac{1}{2}} \epsilon_k + C_{g_k}^2B^{\frac{2}{3}}\big(\frac{\log |\calN_{\epsilon_k}|}{P_k}\big)^{\frac{2}{3}}$, where the last inequality again follows from \eqref{choices of parameters}.

Using \eqref{Generalization error of input at N plus 1 1 appendix lemma 5}, we can further obtain
\begin{eqnarray}
    \label{Upper bound of the estimated output and ground truth output}
    &\!\!\!\! \!\!\!\!&\E[(\wh \vy_{N+1}(k) - g_k(\vz_{N+1}) )^2] \nonumber\\
    &\!\!\!\!=\!\!\!\!& \E[( \wt g_k^{(L)} - g_k(\vz_{N+1}) )^2]\nonumber\\
    &\!\!\!\!\leq \!\!\!\!&  2\E[( {\vphi_{k}^{N+1}}^\top \wh\va_{k} - g_k(\vz_{N+1}) )^2] + 2\E[( \wt g_k^{(L)} - {\vphi_{k}^{N+1}}^\top \wh\va_{k} )^2]\nonumber\\
    &\!\!\!\!\leq \!\!\!\!&  2\E[( {\vphi_{k}^{N+1}}^\top \wh\va_{k} - g_k(\vz_{N+1}) )^2] + 2( \wt g_k^{(L)} - {\vphi_{k}^{N+1}}^\top \vb_k^{(L)} )^2\nonumber\\
    &\!\!\!\!\lesssim \!\!\!\!& \frac{C_{g_k}^2}{N^2} + \sqrt{\frac{\log N}{N}}C_{g_k}^2   + \frac{P_kC_{g_k}^2}{L} + C_{g_k}^2 B^{\frac{1}{2}} \epsilon_k + C_{g_k}^2B^{\frac{2}{3}}\bigg(\frac{\log |\calN_{\epsilon_k}|}{P_k}\bigg)^{\frac{2}{3}}\nonumber\\
    &\!\!\!\!\lesssim \!\!\!\!& \sqrt{\frac{\log N}{N}}C_{g_k}^2   + \frac{P_kC_{g_k}^2}{L} + C_{g_k}^2B^{\frac{2}{3}}\bigg(\frac{\log |\calN_{\epsilon_k}|}{P_k}\bigg)^{\frac{2}{3}},
\end{eqnarray}
where \eqref{upper bound of output at N+1 in the transformer} is used in the third inequality and we use $\epsilon_k\leq B^{-\frac{1}{2}}(\sqrt{\frac{\log N}{N}} +  \frac{P_k}{L})$ in the last line.

Finally, note that the covering number of continuous degrees of freedom is $O\big( \frac{3\sqrt{B}}{\epsilon_k} \big)^{nD}$. Together with the condition $\epsilon_k \le B^{-\frac{1}{2}}\big(\sqrt{\frac{\log N}{N}} + \frac{P_k}{L}\big)$, this implies that $P_k = \Omega(B \log |\calN_{\epsilon_k}| ) = \Omega(nDB\log(\frac{B}{\sqrt{(\log N)/N} +  P_k/L})  )$.  When $\sqrt{(\log N)/N}\geq P_k/L$, this yields $P_k = \Omega(nDB\log(B\sqrt{N}/\sqrt{\log N}) ) $. When $\sqrt{(\log N)/N}< P_k/L$, the bound instead reduces to $P_k = \Omega(nDB\log(BL) ) $. Thus one may choose $P_k = \Omega(nDB\max\{\log(BL), \log(B\sqrt{N}/\sqrt{\log N}) \}) $.

With this choice, the third term in \eqref{Upper bound of the estimated output and ground truth output} becomes bounded by a constant when $P_k$ is sufficiently large. Consequently, \eqref{Upper bound of the estimated output and ground truth output} reduces to
\begin{eqnarray}
    \label{Upper bound of the estimated output and ground truth output continuous}
    \E[(\wh \vy_{N+1}(k) - g_k(\vz_{N+1}) )^2] = O\bigg(\sqrt{\frac{\log N}{N}} + \frac{nDB\max\{\log(BL), \log(B\sqrt{N}/\sqrt{\log N}) \}}{L} \bigg).
\end{eqnarray}
To ensure that the second term is dominated by the statistical error term, it suffices to require
\begin{eqnarray}
    \label{Sufficient condition for output continuous}
    \sqrt{\frac{\log N}{N}} \geq \frac{nDB\max\{\log(BL), \log(B\sqrt{N}/\sqrt{\log N}) \}}{L}.
\end{eqnarray}
This condition implies $L \geq \frac{nDB\sqrt{N}}{\sqrt{\log N}}\max\{\log(BL), \log(B\sqrt{N}/\sqrt{\log N})\} $. Since such a choice necessarily satisfies $L\geq \sqrt{N}/\sqrt{\log N}$, it follows that $\E[(\wh \vy_{N+1}(k) - g_k(\vz_{N+1}) )^2] = O\big(\sqrt{\frac{\log N}{N}}\big)$ whenever $L \geq \frac{nDB\log(BL)\sqrt{N}}{\sqrt{\log N}}  $.

To derive a simpler sufficient condition, we additionally consider the case where $L\geq B$. Under this condition, $\frac{2nDB\sqrt{N}}{\sqrt{\log N}}\log L\geq \frac{nDB\sqrt{N}}{\sqrt{\log N}}(\log L + \log B )\geq \frac{nDB\sqrt{N}}{\sqrt{\log N}}\log(BL)$. Therefore, it suffices to require
\begin{eqnarray}
    \label{sufficient condition for continuous case}
    \frac{L}{\log L} \geq \frac{2nDB\sqrt{N}}{\sqrt{\log N}},
\end{eqnarray}
which guarantees  $\E[(\wh \vy_{N+1}(k) - g_k(\vz_{N+1}) )^2] = O\big(\sqrt{\frac{\log N}{N}}\big)$.

For discrete degrees of freedom, we first normalize $(x_{1}, \dots, x_{n})$ by $\sqrt{n(d-1)}$ in the embedding matrix so as to compress the input domain. Under this normalization, the discrete domain can be embedded into the continuous region characterized by $\sum_{i=1}^{n}\frac{ x_i^2}{n(d-1)} \le d-1$. Consequently, the covering number associated with this region satisfies $O\big( \frac{3\sqrt{d-1}}{\epsilon_k} \big)^{n}$. Following the same argument as in the continuous case and applying \eqref{Upper bound of the estimated output and ground truth output}, we can choose $P_k = \Omega(nd\max\{\log(dL), \log(d\sqrt{N}/\sqrt{\log N}) \} )$ which yields the following generalization bound for discrete degrees of freedom:
\begin{eqnarray}
    \label{Upper bound of the estimated output and ground truth output discrete}
    \E[\|\wh \vy_{N+1} -  \vd(\{x_{N+1,j}\}_{j=1}^{n} )\|_2^2]= O\bigg( \sqrt{\frac{\log N}{N}} + \frac{nd\max\{\log(dL), \log(d\sqrt{N}/\sqrt{\log N}) \}}{L} \bigg).
\end{eqnarray}
Note that although the above bound is derived based on the covering number of the corresponding continuous domain, the discrete system can be regarded as a special case obtained by discretizing this domain. Therefore, the same generalization bound also applies to the discrete setting. Following the same derivation as that of \eqref{sufficient condition for continuous case}, we additionally assume that $L\ge d$. Then, a sufficient condition is
\begin{eqnarray}
    \label{sufficient condition for discrete case}
    \frac{L}{\log L} \geq \frac{2nd\sqrt{N}}{\sqrt{\log N}},
\end{eqnarray}
which guarantees  $\E[\|\wh \vy_{N+1} -  \vd(\{x_{N+1,j}\}_{j=1}^{n} )\|_2^2] = O\big(\sqrt{\frac{\log N}{N}}\big)$.

\end{proof}

\section{Proof of \Cref{feature_construction_of_functions}}
\label{Proof of feature_construction_of_functions}

The proof follows the general strategy of \cite[Lemma 1]{li2025transformers}, but extends the result from the unit ball to a general bounded domain of radius $B$. Consequently, our analysis recovers \cite[Lemma 1]{li2025transformers} as the special case $B=1$.
For notational convenience, we express the Fourier transform $F_{g_k}(\boldsymbol{\omega})$ in polar form as $F_{g_k}(\boldsymbol{\omega}) = |F_{g_k}(\boldsymbol{\omega})|e^{j \theta_{g_k}(\boldsymbol{\omega}) }$, where $\theta_{g_k}(\boldsymbol{\omega})$ denotes the phase. Under this representation, the function difference ${g_k}(\vz) - {g_k}({\bf 0})$ admits the integral representation ${g_k}(\vz) - {g_k}({\bf 0}) = \int_{\vw\neq {\bf 0}} \big(  \cos(\boldsymbol{\omega}^\top \vz + \theta_{g_k}(\boldsymbol{\omega})  )  - \cos(\theta_{g_k}(\boldsymbol{\omega})  )  \big) |F_{g_k}(\boldsymbol{\omega})| \,d\boldsymbol{\omega}$. To approximate ${g_k}(\vz) - {g_k}({\bf 0})$, we construct a finite set of features based on the smoothed step function in \eqref{smoothed_step_function}. For this purpose, we first introduce a quantity that measures the approximation discrepancy between $\phi(\tau_k x)$ and the unit step function $\mathbf{1}(x>0)$:
\begin{eqnarray}
    \label{The definition of delta function}
\delta_{\tau_k}\;:=\;\inf_{0 < \epsilon \leq \tfrac{1}{2}}\bigg(2\epsilon\;+\;\sup_{|x| \geq \epsilon}\big| \phi(\tau_k x) - \mathbf{1}(x > 0) \big|\bigg).
\end{eqnarray}

According to \Cref{approximation theorem in the continuous case}, for any $g_k \in \calG$, any $\vz$ satisfying $\|\vz\|_2 \le \sqrt{B}$, any positive integer $m$, and any $\tau_k > 2$, we have
\begin{eqnarray}
    \label{The upper bound of approximation term g_k another}
    \big| {g_k}(\vz) - {g_k}({\bf 0}) - g^{\mathrm{approx}}_k (\vz) \big|
    \le C_{g_k} \sqrt{B} \Bigl( 3\delta_{\tau_k} + \frac{\pi}{m} \Bigr),
\end{eqnarray}
where
\begin{eqnarray}
    \label{The definition of approximate function g another}
    g^{\mathrm{approx}}_k (\vz) := \sum_{i=1}^m \int_{\boldsymbol{\omega} \ne {\bf 0}} \int_{t} A_{g_k}(i,t,\boldsymbol{\omega})\,\phi\! \big( \tau_k(\|\boldsymbol{\omega}\|_2^{-1} \boldsymbol{\omega}^\top \vz - t ) \big) \Lambda(i, \mathrm{d}t, \mathrm{d}\boldsymbol{\omega}),
\end{eqnarray}
and $\Lambda(i, \mathrm{d}t, \mathrm{d}\boldsymbol{\omega})$ is a probability measure on $[m] \times \R \times \R^{nD}$ independent of $g_k$, while $A_{g_k}(i,t,\boldsymbol{\omega})$ is a weight function depending on $g_k,i,t,\boldsymbol{\omega}$.

To obtain a finite-sum representation, we approximate the integral in \eqref{The definition of approximate function g another} via random sampling. Specifically, we draw $P_k$ independent samples $(i_j, t_j, \boldsymbol{\omega}_j), 1\leq j \leq P_k$, from $\Lambda(i, t, \boldsymbol{\omega})$. Define
\begin{eqnarray}
    \label{The definition of C minus gk}
    \wt C_{g_k} = C_{g_k} - \sup_{\wt g_k \in \calG} |\wt g_k({\bf 0})|.
\end{eqnarray}
For any fixed $(\wh g_k, \wh\vz)\in\calN_{\epsilon_k}$, we have
\begin{eqnarray}
    \label{Conditions for Berstein inequality}
    \E[ A_{\wh g_k}(i_j,t_j,\boldsymbol{\omega}_j)\,\phi\! \big( \tau_k(\|\boldsymbol{\omega}_j\|_2^{-1} \boldsymbol{\omega}_j^\top \wh \vz - t_j ) \big) ] &\!\!\!\!=\!\!\!\!& \wh g^{\mathrm{approx}}_k (\wh\vz),\\
    | A_{\wh g_k}(i_j,t_j,\boldsymbol{\omega}_j)\,\phi\! \big( \tau_k(\|\boldsymbol{\omega}_j\|_2^{-1} \boldsymbol{\omega}_j^\top \wh \vz - t_j ) \big) | &\!\!\!\!\leq\!\!\!\!& 3m \wt C_{g_k},\\
    \mathrm{Var}[A_{\wh g_k}(i_j,t_j,\boldsymbol{\omega}_j)\,\phi\! \big( \tau_k(\|\boldsymbol{\omega}_j\|_2^{-1} \boldsymbol{\omega}_j^\top \wh \vz - t_j ) \big) ] &\!\!\!\!\leq \!\!\!\!&\E[( A_{\wh g_k}(i_j,t_j,\boldsymbol{\omega}_j)\,\phi\! \big( \tau_k(\|\boldsymbol{\omega}_j\|_2^{-1} \boldsymbol{\omega}_j^\top \wh \vz - t_j ) \big))^2 ]\nonumber\\
    &\!\!\!\!\leq \!\!\!\!& \hspace{-0.5cm}\sup_{1\leq i\leq m,t\in\R, \boldsymbol{\omega}\in\R^{nD} } \hspace{-0.8cm} |A_{\wh g_k}(i,t,\boldsymbol{\omega})|\E[|A_{\wh g_k}(i,t,\boldsymbol{\omega})|]\leq 9m \wt C_{g_k}^2,
\end{eqnarray}
where the last inequality follows \Cref{approximation theorem in the continuous case}. Applying the Bernstein inequality together with a union bound over the covering set $\calN_{\epsilon_k}$, we obtain that with high probability the following inequality holds simultaneously for all $(\wh g_k,\wh\vz)\in\calN_{\epsilon_k}$:
\begin{eqnarray}
    \label{Bernstein inequality1}
    \bigg| \wh g^{\mathrm{approx}}_k (\wh\vz) \! - \! \frac{1}{P_k}\sum_{j=1}^{P_k} A_{\wh g_k}(i_j,t_j,\boldsymbol{\omega}_j)\,\phi\! \big( \tau_k(\|\boldsymbol{\omega}_j\|_2^{-1} \boldsymbol{\omega}_j^\top \wh \vz - t_j ) \big)   \bigg| \lesssim \max\bigg\{ \sqrt{\frac{m\wt C_{g_k}^2\log|\calN_{\epsilon_k}|}{P_k}}, \frac{m\wt C_{g_k}\log|\calN_{\epsilon_k}|}{P_k}      \bigg\}.
\end{eqnarray}
Applying the Bernstein inequality with $| A_{\wh g_k}(i_j,t_j,\boldsymbol{\omega}_j) | \leq 3m \wt C_{g_k}$ and $\mathrm{Var}[| A_{\wh g_k}(i_j,t_j,\boldsymbol{\omega}_j) | ]\leq  \E[| A_{\wh g_k}(i_j,t_j,\boldsymbol{\omega}_j) |^2] \leq 9m \wt C_{g_k}^2$ yields that with high probability
\begin{eqnarray}
    \label{Bernstein inequality2}
    \bigg| \E_{(i,t,\boldsymbol{\omega})\sim\Lambda }[| A_{\wh g_k}(i,t,\boldsymbol{\omega}) | ]  -  \frac{1}{P_k}\sum_{j=1}^{P_k}| A_{\wh g_k}(i_j,t_j,\boldsymbol{\omega}_j) |   \bigg| \lesssim \max\bigg\{ \sqrt{\frac{m\wt C_{g_k}^2\log|\calN_{\epsilon_k}|}{P_k}}, \frac{m\wt C_{g_k}\log|\calN_{\epsilon_k}|}{P_k}      \bigg\}.
\end{eqnarray}

Suppose $P_k \gtrsim m\log|\calN_{\epsilon_k}|$. Then the first term dominates the second one, and the bound simplifies to
\begin{eqnarray}
    \label{Upper bound of summation of approximation function}
    \frac{1}{P_k}\sum_{j=1}^{P_k}| A_{\wh g_k}(i_j,t_j,\boldsymbol{\omega}_j) | \lesssim \E_{(i,t,\boldsymbol{\omega})\sim\Lambda }[ | A_{\wh g_k}(i,t,\boldsymbol{\omega})| ] + \sqrt{\frac{m\wt C_{g_k}^2\log|\calN_{\epsilon_k}|}{P_k}} \leq 4\wt C_{g_k}\leq 4 C_{g_k},
\end{eqnarray}
where the second inequality uses \Cref{approximation theorem in the continuous case}. Moreover, it follows that
\begin{eqnarray}
    \label{Upper bound of summation of approximation function difference}
    &\!\!\!\!\!\!\!\!&\bigg| \frac{1}{P_k}\sum_{j=1}^{P_k}\bigg( A_{\wh g_k}(i_j,t_j,\boldsymbol{\omega}_j)\,\phi\! \big( \tau_k(\|\boldsymbol{\omega}_j\|_2^{-1} \boldsymbol{\omega}_j^\top \wh \vz - t_j ) \big)  -  A_{\wh g_k}(i_j,t_j,\boldsymbol{\omega}_j)\,\phi\! \big( \tau_k(\|\boldsymbol{\omega}_j\|_2^{-1} \boldsymbol{\omega}_j^\top \vz - t_j ) \big) \bigg) \bigg|\nonumber\\
    &\!\!\!\!\leq\!\!\!\!& \frac{1}{P_k}\sum_{j=1}^{P_k}|A_{\wh g_k}(i_j,t_j,\boldsymbol{\omega}_j)|\cdot \big|\tau_k  \|\boldsymbol{\omega}_j\|_2^{-1} \boldsymbol{\omega}_j^\top (\wh \vz - \vz)\big|\nonumber\\
    &\!\!\!\!\leq\!\!\!\!& \frac{1}{P_k}\sum_{j=1}^{P_k}|A_{\wh g_k}(i_j,t_j,\boldsymbol{\omega}_j)| \tau_k \|\wh \vz - \vz\|_2 \leq 4 C_{g_k}\tau_k \epsilon_{k}.
\end{eqnarray}

Furthermore, \eqref{Upper bound of summation of approximation function} implies $| g_k({\bf 0})| +  \frac{1}{P_k}\sum_{j=1}^{P_k}| A_{\wh g_k}(i_j,t_j,\boldsymbol{\omega}_j) | \lesssim  4 C_{g_k}$. Collecting the preceding bounds yields
\begin{eqnarray}
    \label{Upper bound of final conclusion 1 Lemma 1}
    &\!\!\!\!\!\!\!\!&\bigg|{g_k}(\vz) - {g_k}({\bf 0}) - \frac{1}{P_k}\sum_{j=1}^{P_k}  A_{\wh g_k}(i_j,t_j,\boldsymbol{\omega}_j) \,\phi\! \big( \tau_k(\|\boldsymbol{\omega}_j\|_2^{-1} \boldsymbol{\omega}_j^\top \vz - t_j ) \big) \bigg|\nonumber\\
    &\!\!\!\!\leq\!\!\!\!& \bigg|{g_k}(\vz) - {g_k}({\bf 0}) - \wh g_k(\vz) + \wh g_k({\bf 0})\bigg| + \bigg| \wh g_k(\vz) - \wh g_k({\bf 0}) - \wh g^{\mathrm{approx}}_k (\vz)\bigg| \nonumber\\
    &\!\!\!\!\!\!\!\!& + \bigg|\wh g^{\mathrm{approx}}_k (\vz) - \frac{1}{P_k}\sum_{j=1}^{P_k} A_{\wh g_k}(i_j,t_j,\boldsymbol{\omega}_j)\,\phi\! \big( \tau_k(\|\boldsymbol{\omega}_j\|_2^{-1} \boldsymbol{\omega}_j^\top \wh \vz - t_j ) \big)    \bigg|  \nonumber\\
    &\!\!\!\!\!\!\!\!& + \bigg| \frac{1}{P_k}\sum_{j=1}^{P_k}\bigg( A_{\wh g_k}(i_j,t_j,\boldsymbol{\omega}_j)\,\phi\! \big( \tau_k(\|\boldsymbol{\omega}_j\|_2^{-1} \boldsymbol{\omega}_j^\top \wh \vz - t_j ) \big)  -  A_{\wh g_k}(i_j,t_j,\boldsymbol{\omega}_j)\,\phi\! \big( \tau_k(\|\boldsymbol{\omega}_j\|_2^{-1} \boldsymbol{\omega}_j^\top \vz - t_j ) \big) \bigg) \bigg|\nonumber\\
    &\!\!\!\! \lesssim\!\!\!\!& C_{g_k} \bigg(\epsilon_{k}  + \sqrt{B} ( 3\delta_{\tau_k} + \pi/m ) + \sqrt{m \log|\calN_{\epsilon_k}|/P_k} + 4 \tau_k \epsilon_{k}\bigg)\nonumber\\
    &\!\!\!\! \lesssim\!\!\!\!& C_{g_k} \bigg( \sqrt{B}/\tau_k  + B^{\frac{1}{3}}( \log|\calN_{\epsilon_k}|/P_k)^{\frac{1}{3}} + \tau_k \epsilon_{k} \bigg).
\end{eqnarray}
Here the second inequality follows from
\Cref{approximation theorem in the continuous case},
\eqref{Bernstein inequality1},
and \eqref{Upper bound of summation of approximation function difference}.
The final bound is obtained by choosing $m = \lceil B^{\frac{1}{3}}  (\frac{P_k}{\log|\calN_{\epsilon_k}|})^{\frac{1}{3}} \rceil$ and using $\delta_{\tau_k}\leq \frac{1}{\tau_k}$ for any $\tau_k>2$ \cite{li2025transformers}. Setting $a_{k, j}^\star \coloneqq  A_{\wh g_k}(i_j,t_j,\boldsymbol{\omega}_j)$ completes the proof.

\section{Proof of \Cref{lemma: transformer emulate proximal GD}}
\label{Proof of transformer emulate proximal GD}

In this proof, we construct a Transformer to approximate the proximal gradient descent iterations described above. We begin by specifying the input structure for each layer of the Transformer. Compared with the construction in \cite[Lemma 4]{li2025transformers}, the main difference lies in the design of the input embedding matrix. Since our Transformer simultaneously approximates two output functions, namely $g_1$ and $g_2$, we augment the embedding with additional coordinates to encode the corresponding feature vectors, regression coefficients, and auxiliary variables required by the analysis. To accommodate the two-output setting, we define the input matrix for the $\ell$-th hidden layer ($0\leq \ell\leq L$) as
\begin{eqnarray}
    \label{the construction of embedding vector each layer}
    \mH^{(\ell)} := \begin{bmatrix}\vz_{1}^{(\ell)} & \cdots & \vz_{N}^{(\ell)} & \vz_{N+1}^{(\ell)} \\ g_{1,1}^{(\ell)} & \cdots & g_{1,N}^{(\ell)} & 0 \\  g_{2,1}^{(\ell)} & \cdots & g_{2,N}^{(\ell)} & 0 \\  w_{1}^{(\ell)} & \cdots & w_{N}^{(\ell)} & w_{N+1}^{(\ell)} \\ \vphi_{1,1}^{(\ell)} & \cdots & \vphi_{1,N}^{(\ell)} & \vphi_{1,N+1}^{(\ell)} \\ \vphi_{2,1}^{(\ell)} & \cdots & \vphi_{2,N}^{(\ell)} & \vphi_{2,N+1}^{(\ell)} \\ \vb_{1}^{(\ell)} & \cdots & \vb_{1}^{(\ell)} & \vb_{1}^{(\ell)} \\ \vb_{2}^{(\ell)} & \cdots & \vb_{2}^{(\ell)} & \vb_{2}^{(\ell)} \\ \lambda_1^{(\ell)} & \cdots & \lambda_1^{(\ell)} & \lambda_1^{(\ell)}  \\ \lambda_2^{(\ell)} & \cdots & \lambda_2^{(\ell)} & \lambda_2^{(\ell)} \\ \wt g_1^{(\ell)} & \cdots & \wt g_1^{(\ell)} & \wt g_1^{(\ell)} \\ \wt g_2^{(\ell)} & \cdots & \wt g_2^{(\ell)} & \wt g_2^{(\ell)} \end{bmatrix}\in \R^{(nD + 2P_1 + 2P_2 + 12) \times (N+1)},
\end{eqnarray}
where $\vz_{i}^{(\ell)}\in\R^{nD+1}$, $g_{1,i}^{(\ell)},g_{2,i}^{(\ell)}, w_{i}^{(\ell)}, \lambda_1^{(\ell)}, \lambda_2^{(\ell)}, \wt g_1^{(\ell)}, \wt g_2^{(\ell)}\in\R$,  $\vphi_{1,i}^{(\ell)}, \vb_{1}^{(\ell)}\in\R^{P_1+1}$ and $\vphi_{2,i}^{(\ell)},  \vb_{2}^{(\ell)}\in\R^{P_2+1}$. Note that the last rows consist of $N+1$ identical copies of $\lambda_1^{(\ell)}, \lambda_2^{(\ell)}, \wt g_1^{(\ell)}, \wt g_2^{(\ell)}, \vb_{1}^{(\ell)}, \vb_{2}^{(\ell)}$. Compared with the practical embedding in \eqref{the construction of embedding vector}, the augmented embedding matrix in \eqref{the construction of embedding vector each layer} introduces additional coordinates solely for the theoretical analysis, making intermediate quantities explicit in the proof. These auxiliary coordinates are not required in the practical implementation, where the corresponding quantities are generated implicitly by the hidden representations of the Transformer. Consequently, the practical embedding consists only of the physical variables together with a small number of channel-indicator coordinates. Specifically, two indicator coordinates are used for complex-valued wave functions to distinguish the real and imaginary channels, whereas only one indicator coordinate is required for real-valued wave functions.

For the initial layer ($\ell=0$), the elements of $\mH^{(0)}$ can be initialized as
\begin{eqnarray}
    \label{initial values for embedding matrix}
    \vz_{i}^{(0)} = \begin{bmatrix} \vz_i \\ 1  \end{bmatrix}, \vphi_{1,i}^{(0)} = {\bf 0}, \vphi_{2,i}^{(0)} = {\bf 0}, && 1\leq i \leq N+1, \nonumber\\
    g_{1,i}^{(0)} = g_1(\vz_i), g_{2,i}^{(0)}  = g_2(\vz_i),  w_i^{(0)} = 1, && 1\leq i \leq N, \\
    w_{N+1}^{(0)} = 0, \vb_1^{(0)} = {\bf 0}, \vb_2^{(0)} = {\bf 0}, \lambda_1^{(0)} =0, \lambda_2^{(0)} =0, \wt g_1^{(0)} = 0, \wt g_2^{(0)} = 0. && \nonumber
\end{eqnarray}
With this initialization, the construction in \eqref{the construction of embedding vector each layer} is fully specified.

Next, we construct a single Transformer and analyze each of its components. The Transformer consists of $L$ layers. The first layer comprises an attention layer, denoted as {\bf Attn0}, followed by a feed-forward layer, denoted as {\bf FF0}. For the remaining $L-1$ layers (assuming $L-1$ is even), we partition them into $\frac{L-1}{2}$ blocks, where each block consists of two consecutive layers: {\bf Attn1}, {\bf FF1}, {\bf Attn2} and {\bf FF2}. To streamline notation, for $\ell = 1, \dots, L$ we define the intermediate outputs of the attention and feed-forward layers as
\begin{eqnarray}
    \label{definition of attention layer output with multiheads each layer}
    \mH^{(\ell - \frac{1}{2})} &\!\!\!\!=\!\!\!\!& \textup{Attn}_{\mTheta_\textup{attn}^{(\ell)}}(\mH^{(\ell - 1)}),\\
    \label{definition of MLP layer each layer}
    \mH^{(\ell - 1)} &\!\!\!\!=\!\!\!\!& \textup{FF}_{\mTheta_\textup{ff}^{(\ell)}}(\mH^{(\ell - \frac{1}{2})}).
\end{eqnarray}
We now proceed to construct these six layers explicitly and explain the rationale behind each design choice. For notational convenience, we set $E  \coloneqq nD + 2P_1 + 2P_2 + 12$.
\begin{itemize}

\item{{\bf Attn0}:} In this attention layer, we set $M=1$ and initialize  $\mQ_1^{(1)} = \mK_1^{(1)} = \mV_1^{(1)} = {\bf 0}$. Hence, the output of {\bf Attn0} is simply $\mH^{(\frac{1}{2})} = \mH^{(0)}$.

\item{{\bf FF0}:} In this feed-forward layer, we update $\vphi_{k,i}^{(\ell)}, k=1,2, 1\leq i \leq N+1$ and $\vb_{k}^{(\ell)}, k =1,2$ as
\begin{eqnarray}
    \label{The output of FF0 1}
    \vphi_{k,i}^{(\ell)} &\!\!\!\!=\!\!\!\!& \vphi_{k,i}^{(\ell-\frac{1}{2})} +  \begin{bmatrix} \phi_{k,1}(\vz_i)  & \cdots & \phi_{k,P_k}(\vz_i) & 1    \end{bmatrix}^\top,\\
    \label{The output of FF0 2}
    \lambda_{k}^{(\ell)} &\!\!\!\!=\!\!\!\!& \ol\lambda_{k},
\end{eqnarray}
where $\ol\lambda_{k} = c\bigg(\big(\frac{\log N}{N} \big)^{\frac{1}{6}}C_{g_k}^{-\frac{1}{3}}(\sqrt{\frac{\log N}{N}}C_{g_k}^2 + \eta_k^2 + \frac{P_kC_{g_k}^2}{L}    )^{\frac{2}{3}} + \sqrt{\frac{\log N}{N}}C_{g_k} + C_{g_k}^{-1} \eta_k^2\bigg)$ for some constants $c$.

Since  $\phi(z) = \sigma_{\textup{ff}}(z + \frac{1}{2} ) - \sigma_{\textup{ff}}(z - \frac{1}{2}  )$, we have $\phi_{k,i}(\vz) = \sigma_{\textup{ff}}(\tau_k  ( \frac{\vw_{k,i}^\top \vz}{\|\vw_{k,i}\|_2}  - t_{k,i}    ) + \frac{1}{2} ) - \sigma_{\textup{ff}}(\tau_k  ( \frac{\vw_{k,i}^\top \vz}{\|\vw_{k,i}\|_2}  - t_{k,i}    ) - \frac{1}{2} )$, where $\tau_k =\frac{B^{\frac{1}{4}}}{\sqrt{\epsilon_k}}$. Furthermore, noting that $1 = \sigma_{\textup{ff}}(1) - \sigma_{\textup{ff}}( -1 )$, we construct $\mW^{(1)}, \mU^{(1)}\in\R^{E\times E}$ as
\begin{eqnarray}
    \label{The definition of W}
    &\!\!\!\!\!\!\!\!&\mW^{(1)}(1:2P_1+2P_2+4, 1:nD+1 ) = \begin{bmatrix} \|\vw_{1,1}\|_2^{-1}\tau_1 \vw_{1,1}^\top   &  -t_{1,1} \tau_1 +\frac{1}{2} \\  \vdots & \vdots \\ \|\vw_{1,P_1}\|_2^{-1}\tau_1 \vw_{1,P_1}^\top   &  -t_{1,P_1} \tau_1 +\frac{1}{2} \\ {\bf 0}^\top & 1 \\ \|\vw_{1,1}\|_2^{-1}\tau_1 \vw_{1,1}^\top   &  -t_{1,1} \tau_1 -\frac{1}{2} \\  \vdots & \vdots \\ \|\vw_{1,P_1}\|_2^{-1}\tau_1 \vw_{1,P_1}^\top   &  -t_{1,P_1} \tau_1 -\frac{1}{2} \\ {\bf 0}^\top & -1 \\ \|\vw_{2,1}\|_2^{-1}\tau_2 \vw_{2,1}^\top   &  -t_{2,1} \tau_2 +\frac{1}{2} \\  \vdots & \vdots \\ \|\vw_{2,P_2}\|_2^{-1}\tau_2 \vw_{2,P_2}^\top   &  -t_{2,P_2} \tau_2 +\frac{1}{2} \\ {\bf 0}^\top & 1 \\ \|\vw_{2,1}\|_2^{-1}\tau_2 \vw_{2,1}^\top   &  -t_{2,1} \tau_2 -\frac{1}{2} \\  \vdots & \vdots \\ \|\vw_{2,P_2}\|_2^{-1}\tau_2 \vw_{2,P_2}^\top   &  -t_{2,P_2} \tau_2 -\frac{1}{2} \\ {\bf 0}^\top & -1  \end{bmatrix}, \nonumber\\
    &\!\!\!\!\!\!\!\!&\mW^{(1)}(2P_1 + 2P_2 + 5,  :) = \ol\lambda_1 \ve_{nD+1}^\top \qquad \text{and} \qquad \mW^{(1)}(2P_1 + 2P_2 + 6,  :) = \ol\lambda_2 \ve_{nD+1}^\top.
\end{eqnarray}
\begin{eqnarray}
    \label{The definition of U}
    &\!\!\!\!\!\!\!\!& \mU^{(1)}(nD + 5: nD +P_1+ 5, 1:2P_1 + 2 ) = \begin{bmatrix}  \mId_{P_1+1} & -\mId_{P_1+1} \end{bmatrix},\nonumber\\
    &\!\!\!\!\!\!\!\!& \mU^{(1)}(nD + P_1 +6: nD + P_1 + P_2 + 6, 2P_1 + 3 : 2P_1 + 2P_2 +4 ) = \begin{bmatrix}  \mId_{P_2+1} & -\mId_{P_2+1} \end{bmatrix}\nonumber\\
    &\!\!\!\!\!\!\!\!& \mU^{(1)}( nD+2P_1+2P_2+9,:) = \ve_{2P_1+2P_2+5}^\top,\nonumber\\
    &\!\!\!\!\!\!\!\!& \mU^{(1)}( nD+2P_1+2P_2+10,:) = \ve_{2P_1+2P_2+6}^\top,
\end{eqnarray}
where $\ve_i$ denotes the $i$-th standard basis vector.

\item{{\bf Attn1}:} This attention layer is designed to emulate the following iteration:
\begin{eqnarray}
    \label{The update of PGM appendix}
    \vb_{k}^{(\ell-\frac{1}{2})} = \vb_{k}^{(\ell-1)} - \frac{2\mu_k}{N}\sum_{i=1}^{N}({\vphi_{k,i}^{(\ell-1)}}^\top \vb_{k}^{(\ell-1)} - g_{k,i}^{(\ell-1)} )\vphi_{k,i}^{(\ell-1)}  + \vf_k^{(\ell-1)},
\end{eqnarray}
where $\vf_k^{(\ell-1)}$ denotes the residual term.

Notice that the sigmoid function $\sigma_\textup{attn}(\cdot)$ admits the following first-order approximation: $\sigma_\textup{attn}(\tau^{-1} x) = \sigma_\textup{attn}( 0 ) + \frac{x}{4\tau} + O(\frac{x^2}{\tau^2})$.  Consequently, for any $x$, one can write the approximate inversion $x\approx 4\tau(\sigma_\textup{attn}(\tau^{-1} x) - \sigma_\textup{attn}(0) )$. Applying this approximation, the update equation in \eqref{The update of PGM appendix} can then be expressed in the form
\begin{eqnarray}
    \label{The update of PGM appendix1}
    \vb_{k}^{(\ell-\frac{1}{2})} &\!\!\!\!=\!\!\!\!& \vb_{k}^{(\ell-1)} - \frac{2\mu_k}{N}\sum_{i=1}^{N}({\vphi_{k,i}^{(\ell-1)}}^\top \vb_{k}^{(\ell-1)} - g_{k,i}^{(\ell-1)} )\vphi_{k,i}^{(\ell-1)}  + \vf_k^{(\ell-1)}\nonumber\\
    &\!\!\!\! = \!\!\!\!& \vb_{k}^{(\ell-1)} + \frac{2}{N}\sum_{i=1}^{N+1}4\tau_k( \sigma_\textup{attn}(\tau_k^{-1}\mu_k g_{k,i}^{(\ell-1)}   ) - \sigma_\textup{attn}(0) )\vphi_{k,i}^{(\ell-1)}\nonumber\\
    &\!\!\!\! \!\!\!\!& - \frac{2}{N}\sum_{i=1}^{N+1}4\tau_k( \sigma_\textup{attn}(\tau_k^{-1}\mu_k {\vphi_{k,i}^{(\ell-1)}}^\top \vb_{k}^{(\ell-1)}   ) - \sigma_\textup{attn}(0) )\vphi_{k,i}^{(\ell-1)}\nonumber\\
    &\!\!\!\! \!\!\!\!& + \frac{2}{N}\sum_{i=1}^{N+1}4\tau_k( \sigma_\textup{attn}(\tau_k^{-1}\mu_k (1 - w_{i}^{(\ell-1)}) \wt g_k^{(\ell-1)}  ) - \sigma_\textup{attn}(0) )\vphi_{k,i}^{(\ell-1)}, k =1,2.
\end{eqnarray}
Here, $\vf_k^{(\ell-1)}$ represents the discrepancy between the two formulations above. By invoking the bound $|x -  4\tau(\sigma_\textup{attn}(\tau^{-1} x) - \sigma_\textup{attn}(0) )|\leq 2\tau^{-1}|x|^2$ from \cite[eq.(102)]{li2025transformers}, we can further establish the following upper bound:
\begin{eqnarray}
    \label{upper bound of e l in the PGM}
    \|\vf_k^{(\ell)} \|_{\infty}&\!\!\!\!\leq\!\!\!\!& \frac{4\mu_k^2}{\tau_kN}\sum_{i=1}^{N}(g_{k,i}^{(\ell )} )^2 + \frac{4\mu_k^2}{\tau_kN}\sum_{i=1}^{N+1}({\vphi_{k,i}^{(\ell)}}^\top \vb_{k}^{(\ell)})^2 + \frac{4\mu_k^2}{\tau_kN}(\wt g_k^{(\ell)})^2 + \frac{2\mu_k}{N}|\wt g_k^{(\ell)} - (\vphi_{k,N+1}^{(\ell)})^\top \vb_{k}^{(\ell)} |\nonumber\\
    &\!\!\!\!\lesssim \!\!\!\!& \frac{\mu_k^2}{\tau_k}(C_{g_k}^2 + \eta_k^2 + \|\vb_{k}^{(\ell)} \|_1^2 + (\wt g_k^{(\ell)})^2) + \frac{2\mu_k}{N}|\wt g_k^{(\ell)} - (\vphi_{k,N+1}^{(\ell)})^\top \vb_{k}^{(\ell)} |,
\end{eqnarray}
where we have used the estimates $|g_{k,i}^{(\ell )} |\leq |{\vphi_{k}^{i}}^\top\va_{k}^\star| + \eta_k\lesssim \|\vphi_{k}^{i}\|_{\infty}\|\va_{k}^\star \|_1+ \eta_k \lesssim C_{g_k} + \eta_k$ and $|(\vphi_{k,N+1}^{(\ell)})^\top \vb_{k}^{(\ell)}| \leq \|\vphi_{k}^{i}\|_{\infty}\|\vb_k^{(\ell)}\|_1\leq \|\vb_k^{(\ell)}\|_1$ for $\forall i$, applying $\|\vphi_{k}^{i}\|_{\infty} \leq 1$, $\|\va_{k}^\star \|_1\leq 4 C_{g_k}$ and \eqref{final conclusion of Lemma 1}.

To implement the update above, we employ $8$ attention heads in this layer. In particular, we set
\begin{eqnarray*}
    \label{The definition of Q K V 1-1}
    &\!\!\!\!\!\!\!\!& \mQ_1^{(\ell)}(1,:) = \tau_1^{-1}\mu_1\ve_{nD+2}^\top, \quad \mK_1^{(\ell)}(1,:) = \ve_{nD+1}^\top,\\
    &\!\!\!\!\!\!\!\!& \mV_1^{(\ell)}(nD+P_1+P_2+7:nD+2P_1+P_2+7, nD+5: nD+P_1+5 ) = 8\tau_1\mId_{P_1+1}, \\
    \label{The definition of Q K V 1-2}
    &\!\!\!\!\!\!\!\!& \mQ_2^{(\ell)}(1:P_1+1, nD+5: nD+P_1+5) = \tau_1^{-1}\mu_1\mId_{P_1+1},  \\
    &\!\!\!\!\!\!\!\!& \mK_2^{(\ell)}(1:P_1+1,nD+P_1+P_2+7:nD+2P_1+P_2+7) = \mId_{P_1+1}, \quad \mV_2^{(\ell)} = - \mV_1^{(\ell)}, \\
    \label{The definition of Q K V 1-3}
    &\!\!\!\!\!\!\!\!& \mQ_3^{(\ell)}(1,:) = \tau_1^{-1}\mu_1(\ve_{nD+1} - \ve_{nD+4})^\top, \quad \mK_3^{(\ell)}(1.:) = \ve_{nD+2P_1+2P_2+11}^\top, \quad \mV_3^{(\ell)} =  \mV_1^{(\ell)},\\
    \label{The definition of Q K V 1-4}
    &\!\!\!\!\!\!\!\!& \mQ_4^{(\ell)} = {\bm 0}_{E\times E}, \quad \mK_4^{(\ell)} = {\bf 0}_{E\times E}, \quad \mV_4^{(\ell)} =  -\mV_1^{(\ell)},\\
    \label{The definition of Q K V 1-5}
    &\!\!\!\!\!\!\!\!& \mQ_5^{(\ell)}(1,:) = \tau_2^{-1}\mu_2\ve_{nD+3}^\top, \quad \mK_5^{(\ell)}(1,:) = \ve_{nD+1}^\top,\\
    &\!\!\!\!\!\!\!\!& \mV_5^{(\ell)}(nD+2P_1+P_2+8:nD+2P_1+2P_2+8, nD+P_1+6: nD+P_1 + P_2+6) = 8\tau_2\mId_{P_2+1}, \\
    \label{The definition of Q K V 1-6}
    &\!\!\!\!\!\!\!\!& \mQ_6^{(\ell)}(1:P_2+1, nD+P_1+6: nD+P_1 + P_2+6) = \tau_2^{-1}\mu_2\mId_{P_2+1},  \\
    &\!\!\!\!\!\!\!\!& \mK_6^{(\ell)}(1:P_2+1,nD+2P_1+P_2+8:nD+2P_1+2P_2+8) = \mId_{P_2+1}, \quad \mV_6^{(\ell)} = - \mV_5^{(\ell)},\\
    \label{The definition of Q K V 1-7}
    &\!\!\!\!\!\!\!\!& \mQ_7^{(\ell)}(1,:) = \tau_2^{-1}\mu_2(\ve_{nD+1} - \ve_{nD+4})^\top, \quad \mK_7^{(\ell)}(1.:) = \ve_{nD+2P_1+2P_2+12}^\top, \quad \mV_7^{(\ell)} =  \mV_5^{(\ell)},\\
    \label{The definition of Q K V 1-8}
    &\!\!\!\!\!\!\!\!& \mQ_8^{(\ell)} = {\bf 0}_{E\times E}, \quad \mK_8^{(\ell)} = {\bf 0}_{E\times E}, \quad \mV_8^{(\ell)} =  -\mV_5^{(\ell)}.
\end{eqnarray*}

\item{{\bf FF1}:} In this feed-forward layer, our goal is to emulate the update
\begin{eqnarray}
    \label{emulation of inexact PGM in the appendix}
    \vb_{k}^{(\ell )} &\!\!\!\!=\!\!\!\!& \calP_{\mu_k\lambda_k^{(\ell-\frac{1}{2})}}(\vb_{k}^{(\ell-\frac{1}{2})}),\\
    \label{setting zero in the estimated output in the appendix}
    \wt g_k^{(\ell )}&\!\!\!\!=\!\!\!\!& \wt g_k^{(\ell-\frac{1}{2} )} -  \sigma_{\textup{ff}}(\wt g_k^{(\ell-\frac{1}{2} )}  ) + \sigma_{\textup{ff}}(-\wt g_k^{(\ell-\frac{1}{2} )}  ) =0.
\end{eqnarray}
Recalling that  $\calP_{\mu_k\lambda_k }(z) = z + \sigma_{\textup{ff}}(\mu_k\lambda_k ) - \sigma_{\textup{ff}}(z + \mu_k\lambda_k ) + \sigma_{\textup{ff}}(z - \mu_k\lambda_k )$, we can equivalently rewrite \eqref{emulation of inexact PGM in the appendix} as
\begin{eqnarray}
    \label{emulation of inexact PGM in the appendix1}
    \vb_{k}^{(\ell )} &\!\!\!\!=\!\!\!\!& \vb_{k}^{(\ell-\frac{1}{2})} +  \sigma_{\textup{ff}}(\mu_k\lambda_k^{(\ell-\frac{1}{2})}\mathbf{1}_{P_k+1}) - \sigma_{\textup{ff}}(\vb_{k}^{(\ell-\frac{1}{2})} + \mu_k\lambda_k^{(\ell-\frac{1}{2})}\mathbf{1}_{P_k+1})\nonumber\\
    &\!\!\!\!\!\!\!\!&+ \sigma_{\textup{ff}}(\vb_{k}^{(\ell-\frac{1}{2})}- \mu_k\lambda_k^{(\ell-\frac{1}{2})}\mathbf{1}_{P_k+1}).
\end{eqnarray}

To realize the above feed-forward operation within the network, the weight matrices $\mW^{(\ell)}$ and $\mU^{(\ell)}$ can be explicitly constructed as follows. The matrix $\mW^{(\ell)}$ is defined by
\begin{eqnarray}
    \label{The definition of W in the next layer}
    &\!\!\!\!\!\!\!\!&\mW^{(\ell)}(1:2P_1+2P_2+10,:) \nonumber\\
    &\!\!\!\!=\!\!\!\!&\begin{bmatrix} {\bf 0}_{(P_1+1)\times (nD+P_1+P_2+6)} & \mId_{P_1+1} & {\bf 0}_{(P_1+1)\times (P_2+1)} & \mu_1\mathbf{1}_{P_1+1} & {\bf 0} \\ {\bf 0}_{(P_1+1)\times (nD+P_1+P_2+6)} & \mId_{P_1+1} & {\bf 0}_{(P_1+1)\times (P_2+1)} & -\mu_1\mathbf{1}_{P_1+1} & {\bf 0} \\ & & \mu_1\ve_{nD+2P_1+2P_2+9}^\top &  & \\ & & \ve_{nD+2P_1+2P_2+11}^\top &  & \\ & & -\ve_{nD+2P_1+2P_2+11}^\top &  & \\ {\bf 0}_{(P_2+1)\times (nD+2P_1+P_2+8)} & \mId_{P_2+1} & {\bf 0}_{(P_2+1)\times 1} & \mu_2\mathbf{1}_{P_2+1} & {\bf 0} \\ {\bf 0}_{(P_2+1)\times (nD+2P_1+P_2+8)} & \mId_{P_2+1} & {\bf 0}_{(P_2+1)\times 1} & -\mu_2\mathbf{1}_{P_2+1} & {\bf 0} \\ & & \mu_2\ve_{nD+2P_1+2P_2+10}^\top &  & \\ & & \ve_{nD+2P_1+2P_2+12}^\top &  & \\ & & -\ve_{nD+2P_1+2P_2+12}^\top &  & \end{bmatrix},
\end{eqnarray}
Correspondingly, the matrix $\mU^{(\ell)}$ is given by
\begin{eqnarray}
    \label{The definition of U in the next layer}
    &\!\!\!\!\!\!\!\!&\mU^{(\ell)}(nD+P_1+P_2+7:nD+2P_1+P_2+7, 1:2P_1+2P_2+10) = \begin{bmatrix} -\mId_{P_1+1} & \mId_{P_1+1}  & \mathbf{1}_{P_1+1} & {\bf 0}  \end{bmatrix},\nonumber\\
    &\!\!\!\!\!\!\!\!&\mU^{(\ell)}(nD+2P_1+P_2+8:nD+2P_1+2P_2+8, 1:2P_1+2P_2+10)\nonumber\\
    &\!\!\!\!=\!\!\!\!& \begin{bmatrix} {\bf 0}_{(P_2+1)\times (2P_1+5)}  &  -\mId_{P_2+1} & \mId_{P_2+1}  & \mathbf{1}_{P_2+1} & {\bf 0}  \end{bmatrix},\nonumber\\
    &\!\!\!\!\!\!\!\!&\mU^{(\ell)}(nD+2P_1+2P_2+11, 1:2P_1+2P_2+10 ) = \begin{bmatrix} {\bf 0}^\top_{2P_1+3} & -1 & 1  & {\bf 0}^\top_{2P_2+5}  \end{bmatrix},\nonumber\\
    &\!\!\!\!\!\!\!\!&\mU^{(\ell)}(nD+2P_1+2P_2+12, 1:2P_1+2P_2+10 ) = \begin{bmatrix} {\bf 0}^\top_{2P_1+2P_2+8} & -1 & 1    \end{bmatrix}.
\end{eqnarray}

\item{{\bf Attn2}:} This attention layer is designed to update the estimated outputs $\wt g_k$ for $k=1,2$ according to
\begin{eqnarray}
    \label{updating equation of estimated output at N+1}
    \wt g_k^{(\ell + \frac{1}{2})} = (\vphi_{k,N+1}^{(\ell)})^\top \vb_k^{(\ell)} + \wt f_k^{(\ell)},
\end{eqnarray}
where $\wt f_k^{(\ell)}$ is  some residual term. Leveraging the approximation $x\approx 4\tau(\sigma_\textup{attn}(\tau^{-1} x) - \sigma_\textup{attn}(0) )$, the update \eqref{updating equation of estimated output at N+1} can be equivalently expressed as
\begin{eqnarray}
    \label{updating equation of estimated output at N+1 1}
    \wt g_k^{(\ell + \frac{1}{2})} = \wt g_k^{(\ell  )} + \sum_{i=1}^{N+1}4\tau_k (1- w_i^{(\ell)})  ( \sigma_\textup{attn}(\tau_k^{-1}  {\vphi_{k,i}^{(\ell)}}^\top \vb_{k}^{(\ell)}   ) - \sigma_\textup{attn}(0) ).
\end{eqnarray}
Moreover, based on the bound $|x -  4\tau(\sigma_\textup{attn}(\tau^{-1} x) - \sigma_\textup{attn}(0) )|\leq 2\tau^{-1}|x|^2$ and noting that $\wt g_k^{(\ell  )}=0$, the residual $\wt f_k^{(\ell)}$ can be further bounded as
\begin{eqnarray}
    \label{upper bound of approximate error in the attention 2}
    |\wt f_k^{(\ell)}| \leq \frac{2 \big((\vphi_{k,N+1}^{(\ell)})^\top \vb_k^{(\ell)}\big)^2}{\tau_k}.
\end{eqnarray}

To implement \eqref{updating equation of estimated output at N+1 1} within the network, we employ $4$ attention heads in this layer with the following constructions for the query, key, and value matrices:
\begin{eqnarray*}
    \label{design of Q K V the second attention layer}
    &&\mQ_1^{(\ell + 1)}(1:P_1+1,nD+5:nD+P_1+5) = \tau_1^{-1}\mId_{P_1+1},\\
    &&\mK_1^{(\ell+1)}(1:P_1+1, nD+P_1+P_2+7:nD+2P_1+P_2+7) = \mId_{P_1+1},\\
    &&\mV_1^{(\ell+1)}(nD+2P_1+2P_2+11,:) = 4\tau_1(\ve_{nD+1}^\top - \ve_{nD+4}^\top ),\\
    &&\mQ_2^{(\ell+1)} = {\bf 0}, \quad \mK_2^{(\ell+1)} = {\bf 0},  \quad \mV_2^{(\ell+1)} = -\mV_1^{(\ell+1)},\\
    &&\mQ_3^{(\ell + 1)}(1:P_2+1,nD + P_1+6:nD+P_1 + P_2+6) = \tau_2^{-1}\mId_{P_2+1},\\
    &&\mK_3^{(\ell+1)}(1:P_2+1, nD+2P_1+P_2+8:nD+2P_1+2P_2+8) = \mId_{P_2+1},\\
    &&\mV_3^{(\ell+1)}(nD+2P_1+2P_2+12,:) = 4\tau_2(\ve_{nD+1}^\top - \ve_{nD+4}^\top ),\\
    &&\mQ_4^{(\ell+1)} = {\bf 0}, \quad \mK_4^{(\ell+1)} = {\bf 0},  \quad \mV_4^{(\ell+1)} = -\mV_3^{(\ell+1)}.
\end{eqnarray*}

\item{{\bf FF2}:} In this feed-forward layer, no transformation is required, and we simply aim to preserve the hidden representation, i.e., $\mH^{(\ell+1)} = \mH^{(\ell + \frac{1}{2})}$. Consequently, the weight matrices can be trivially set to zero: $\mW^{(\ell+1)} = \mU^{(\ell +1)} = {\bf 0}$.
\end{itemize}
Combining the above analyses, we conclude that it suffices to consider a Transformer with $L$ layers, where the number of attention heads in the $\ell$-th attention layer satisfies
\begin{eqnarray}
    \label{requirement of attenstion layers appendix1}
    M_\ell=
\begin{cases}
1, & \ell=1,\\
4, & \ell>1,\ \ell\ \text{is odd},\\
8, & \ell\ \text{is even}.
\end{cases}
\end{eqnarray}

Now, we first prove that $\vb_k^{(2t+1)} = \va_k^{(t)}$ for all $t\geq 0$. Clearly, $\vb_k^{(1)} = \va_k^{(0)} = {\bf 0}$. Using induction, we assume $\vb_k^{(2t-1)} = \va_k^{(t-1)}$ and aim to show $\vb_k^{(2t+1)} = \va_k^{(t)}$. Specifically, we have
\begin{eqnarray}
    \label{weight at 2t+1 derivation}
    \vb_k^{(2t+1)} &\!\!\!\!=\!\!\!\!& \vb_k^{(2t)} =\calP_{\mu_k\lambda_k^{(2t-\frac{1}{2})}}(\vb_{k}^{(2t-\frac{1}{2})})\nonumber\\
     &\!\!\!\!=\!\!\!\!& \calP_{\mu_k\lambda_k^{(2t-1)}}\bigg(\vb_{k}^{(2t-1)} - \frac{2\mu_k}{N}\sum_{i=1}^{N}({\vphi_{k,i}^{(2t-1)}}^\top \vb_{k}^{(2t-1)} - g_{k,i}^{(2t-1)} )\vphi_{k,i}^{(2t-1)}  + \vf_k^{(2t-1)}\bigg)\nonumber\\
     &\!\!\!\!=\!\!\!\!& \calP_{\mu_k\ol\lambda_k}\bigg( \va_{k}^{(t-1)} - \frac{2\mu_k}{N}\sum_{i=1}^{N}({\vphi_{k}^{i}}^\top \va_{k}^{(t-1)} - g_k(\vz_i) )\vphi_{k}^{i}  \bigg) + \ve_k^{(t)} = \va_{k}^{(t)},
\end{eqnarray}
where $\ve_k^{(t)} = \calP_{\mu_k\lambda_k^{(2t-1)}}(\vb_{k}^{(2t-1)} - \frac{2\mu_k}{N}\sum_{i=1}^{N}({\vphi_{k,i}^{(2t-1)}}^\top \vb_{k}^{(2t-1)} - g_{k,i}^{(2t-1)} )\vphi_{k,i}^{(2t-1)}  + \vf_k^{(2t-1)}) - \calP_{\mu_k\ol\lambda_k}( \va_{k}^{(t-1)} - \frac{2\mu_k}{N}\sum_{i=1}^{N}({\vphi_{k}^{i}}^\top \va_{k}^{(t-1)} - g_k(\vz_i) )\vphi_{k}^{i}  )$, and the second and third equalities use \eqref{emulation of inexact PGM in the appendix} and \eqref{The update of PGM appendix}, respectively.

Using the fact that $\wt g_k^{(2t-1)} = \wt g_k^{(2t-\frac{3}{2})} = (\vphi_{k}^{N+1})^\top  \vb_k^{(2t-2)} + \wt f_k^{(2t-2)} = (\vphi_{k}^{N+1})^\top  \vb_k^{(2t-1)} + \wt f_k^{(2t-2)} = (\vphi_{k}^{N+1})^\top  \va_k^{(t-1)} + \wt f_k^{(2t-2)}$ from \eqref{updating equation of estimated output at N+1}, the bound $|\wt f_k^{(2t-2)}|\leq \frac{2 ((\vphi_{k}^{N+1})^\top \vb_k^{(2t-2)})^2}{\tau_k} = \frac{2 ((\vphi_{k}^{N+1})^\top \vb_k^{(2t-1)})^2}{\tau_k} =\frac{2 ((\vphi_{k}^{N+1})^\top \va_k^{(t-1)})^2}{\tau_k}$ from \eqref{upper bound of approximate error in the attention 2}, and \eqref{upper bound of e l in the PGM}, together with the contraction property of the operator $\calP_{\mu_k\lambda_k^{(2t-1)}}(\cdot)$,  we can then derive
\begin{eqnarray}
    \label{upper bound of residual term in the inexact PGM}
    \|\ve_k^{(t)}\|_{\infty}&\!\!\!\!\leq\!\!\!\!&\|\vf_k^{(2t-1)} \|_{\infty}\nonumber\\
    &\!\!\!\!\lesssim \!\!\!\!& \frac{\mu_k^2}{\tau_k}(C_{g_k}^2 + \eta_k^2 + \|\va_{k}^{(t-1)} \|_1^2 + (\wt g_k^{(2t-1)})^2) + \frac{2\mu_k}{N}|\wt g_k^{(2t-1)} - (\vphi_{k}^{N+1})^\top  \va_{k}^{(t-1)} |\nonumber\\
    &\!\!\!\!\lesssim \!\!\!\!& \frac{\mu_k^2}{\tau_k}(C_{g_k}^2 + \eta_k^2 + \|\va_{k}^{(t-1)} \|_1^2 + ((\vphi_{k}^{N+1})^\top  \va_k^{(t-1)})^2 +  (\wt f_k^{(2t-2)})^2  ) + \frac{\mu_k}{N}|\wt f_k^{(2t-2)}|\nonumber\\
    &\!\!\!\!\lesssim \!\!\!\!& \frac{\mu_k^2}{\tau_k}(C_{g_k}^2 + \eta_k^2 + \|\va_{k}^{(t-1)} \|_1^2 + ((\vphi_{k}^{N+1})^\top  \va_k^{(t-1)})^2 +  \frac{((\vphi_{k}^{N+1})^\top  \va_k^{(t-1)})^4}{\tau_k^2}  ) + \frac{\mu_k((\vphi_{k}^{N+1})^\top \va_k^{(t-1)})^2}{N\tau_k} \nonumber\\
    &\!\!\!\!\lesssim \!\!\!\!& \frac{\mu_k^2}{\tau_k}(C_{g_k}^2 + \eta_k^2 + \|\va_{k}^{(t-1)} \|_1^2 +   \frac{\|\va_{k}^{(t-1)} \|_1^4}{\tau_k^2}  ) + \frac{\mu_k\|\va_{k}^{(t-1)} \|_1^2}{N\tau_k},
\end{eqnarray}
where $(\vphi_{k}^{N+1})^\top  \va_k^{(t-1)}\leq \|\vphi_{k}^{N+1}\|_{\infty} \|\va_k^{(t-1)}\|_1\leq \|\va_k^{(t-1)}\|_1$ is used in the last line.

Next, we bound $\|\va_{k}^{(t)} \|_1$. By \Cref{Performance analysis of PGM}, it follows that
\begin{eqnarray}
    \label{The property of index t update in the PGM1 appendix}
    \|\va_{k}^{(t)}\|_1 \lesssim C_{g_k} + \frac{P_k C_{g_k}^2}{\lambda_k t} + \sqrt{N}(P_k + t)\max_{1\leq j\leq t} \|\ve_k^{(j)}\|_1 + \frac{P_k + t}{\lambda_k}\max_{1\leq j\leq t} (\|\ve_k^{(j)}\|_1 \|\va_{k}^{(j)} \|_1).
\end{eqnarray}
To establish an explicit bound on $\|\va_{k}^{(t)} \|_1$, we proceed by induction to ensure that both $\|\va_{k}^{(t)} \|_1$ and $\|\ve_k^{(t)}\|_1$ satisfy
\begin{eqnarray}
    \label{upper bound of a in the inexact PGM}
    \|\va_{k}^{(t)} \|_1 &\!\!\!\!\leq\!\!\!\!& c_1 P_k \sqrt{N} C_{g_k}, \\
    \label{upper bound of e in the inexact PGM}
    \|\ve_k^{(t)}\|_1 &\!\!\!\!\leq\!\!\!\!& \frac{c_2 C_{g_k}}{(L+P_k)P_kN}
\end{eqnarray}
for some universal constants $c_1$ and $c_2$.

When $t=0$, we have $\|\va_{k}^{(t)} \|_1=0$ and $\|\ve_k^{(t)}\|_1=0$, which clearly satisfy \eqref{upper bound of a in the inexact PGM} and \eqref{upper bound of e in the inexact PGM}. Next, suppose that \eqref{upper bound of a in the inexact PGM} and \eqref{upper bound of e in the inexact PGM} hold for $t = p-1$ and $t = p$, respectively. Then, using \eqref{The property of index t update in the PGM1 appendix} and the lower bound  $\lambda_k \geq c_3\sqrt{\frac{\log N}{N}} C_{g_k} $, it follows that, for $t=p$,
\begin{eqnarray}
    \label{The property of index t update in the PGM1 appendix t=j}
    \|\va_{k}^{(p)}\|_1 \leq c_4\bigg( C_{g_k} + \frac{P_k \sqrt{N} C_{g_k}}{c_3} + c_2 C_{g_k} +  \frac{c_1c_2C_{g_k}}{c_3} + \frac{c_2}{c_3\sqrt{N}P_k} \|\va_{k}^{(p)} \|_1\bigg).
\end{eqnarray}
where $c_3$ and $c_4$ are positive constants. Noting that, if $c_2$ is sufficiently small such that $ \frac{c_2c_4}{c_3 }\leq \frac{1}{4}$, we can absorb the last term into the left-hand side to obtain
\begin{eqnarray}
    \label{The property of index t update in the PGM1 appendix t=j11}
    \|\va_{k}^{(p)}\|_1 \leq   \frac{4c_4}{3}\bigg(C_{g_k} + \frac{ P_k \sqrt{N} C_{g_k}}{c_3} + c_2  C_{g_k} \bigg) +  \frac{c_1C_{g_k}}{3}  \leq c_1 P_k \sqrt{N} C_{g_k}.
\end{eqnarray}
where we have assumed $c_1\geq 2c_4(1+1/c_3+c_2)$.

We next bound $\|\ve_k^{(p+1)}\|_1$.  By \eqref{upper bound of residual term in the inexact PGM}, we have
\begin{eqnarray}
    \label{upper bound of residual term in the inexact PGM appendix 11}
    \|\ve_k^{(p+1)}\|_{\infty}&\!\!\!\!\leq \!\!\!\!& \frac{c_5\mu_k^2}{\tau_k}(C_{g_k}^2 + \eta_k^2 + \|\va_{k}^{(p)} \|_1^2 +   \frac{\|\va_{k}^{(p)} \|_1^4}{\tau_k^2}  ) + \frac{c_5\mu_k\|\va_{k}^{(p)} \|_1^2}{N\tau_k}\nonumber\\
    &\!\!\!\!\leq \!\!\!\!& \frac{c_5 }{4\tau_kP_k^2}(C_{g_k}^2 + \eta_k^2 + c_1^2 P_k^2  N  C_{g_k}^2 +   \frac{c_1^4 P_k^4  N^2  C_{g_k}^4}{\tau_k^2}  ) + \frac{c_5 c_1^2 P_k    C_{g_k}^2}{2\tau_k}\nonumber\\
    &\!\!\!\!\leq \!\!\!\!& \frac{c_2 C_{g_k}}{(L+P_k)P_kN},
\end{eqnarray}
where the second inequality follows from $\mu_k = \frac{1}{2P_k}$ in \Cref{Performance analysis of PGM} and \eqref{The property of index t update in the PGM1 appendix t=j11}. In the last step, we require $\tau_k \geq c_1P_k\sqrt{N}C_{g_k} + \frac{(L+P_k)N}{c_2 C_{g_k}}\big[ \frac{c_5 }{4P_k^2}(C_{g_k}^2 + \eta_k^2 + 2c_1^2 P_k^2  N  C_{g_k}^2) + \frac{c_5 c_1^2 P_k    C_{g_k}^2}{2}  \big] $.

Using \eqref{updating equation of estimated output at N+1} and assuming $\tau_k  \geq c_6 P_k^2  N^2C_{g_k}$ for a sufficiently small constant $c_6$, it follows that the discrepancy between the Transformer's output and the linear prediction satisfies
\begin{eqnarray}
    \label{upper bound of last element in the last layer from the trainsformer}
    |\wh \vy_{N+1}(k) - (\vphi_{k,N+1}^{(L)})^\top \vb_k^{(L)}| &\!\!\!\!=\!\!\!\!& |\wt g_k^{(L)} - (\vphi_{k,N+1}^{(L)})^\top \vb_k^{(L)} |\nonumber\\
    &\!\!\!\!=\!\!\!\!& |\wt g_k^{(L-\frac{1}{2})} - (\vphi_{k}^{N+1})^\top \vb_k^{(L-1)} | = |\wt f_k^{(L-1)}|\nonumber\\
    &\!\!\!\!\leq\!\!\!\!& \frac{2 \big((\vphi_{k}^{N+1})^\top \vb_k^{(L-1)}\big)^2}{\tau_k} = \frac{2 \big((\vphi_{k}^{N+1})^\top \va_k^{(\frac{L-1}{2})}\big)^2}{\tau_k}\nonumber\\
    &\!\!\!\!\leq\!\!\!\!& \frac{2\|\vphi_{k}^{N+1}\|_{\infty}^2\|\va_k^{(\frac{L-1}{2})}\|_1^2}{\tau_k}\lesssim \frac{C_{g_k}}{N},
\end{eqnarray}
where the first inequality follows from \eqref{upper bound of approximate error in the attention 2}.

Finally, by $h(\va_{k}^{(t)}) - h(\va_{k}^\star) \lesssim\frac{P_k\|\va_{k}^\star\|_1^2}{t} + (P_k + t) \max_{1\leq j\leq t} \|\ve_k^{(j)}\|_1 ( \|\va_{k}^{(j)} \|_1 + \eta_k  + C_{g_k} +  \lambda_k )$ in the proof of \cite[Lemma 3]{li2025transformers}, the difference between the PGM-based iterate and the optimal parameter is bounded as
\begin{eqnarray}
    \label{The update of PGM final step appndiex in the fourth lemma}
    h(\vb_k^{(L)}) - h(\va_k^\star) &\!\!\!\!=\!\!\!\!& h(\va_k^{(\frac{L-1}{2})}) - h(\va_k^\star)\nonumber\\
    &\!\!\!\!\lesssim\!\!\!\!&  \frac{P_kC_{g_k}^2}{L} + (L + P_k) \frac{C_{g_k}}{(L+P_k)P_kN} (C_{g_k} + \eta_k+  P_k \sqrt{N} C_{g_k}  + \lambda_k)\nonumber\\
    &\!\!\!\!\lesssim\!\!\!\!&\sqrt{\frac{\log N}{N}}C_{g_k}^2 + \eta_k^2 + \frac{P_kC_{g_k}^2}{L},
\end{eqnarray}
where in the last step we have used the choice $\lambda_k \lesssim \sqrt{\frac{\log N}{N}}C_{g_k}^2 + \eta_k^2 + \frac{P_kC_{g_k}^2}{L} $ in \eqref{choice of lambda in the transformer}.

\section{Proof of \Cref{thm:generalization normalized_transformer_nqs}}
\label{Appendix of generalization of physical Transformer}

We first define $\psi_{\mathrm{c}}$ and $\psi_{\mathrm{d}}$ as the continuous and discrete unnormalized Transformer-based neural quantum states, respectively. We then construct the corresponding normalized Transformer-based neural quantum states as
\begin{eqnarray}
\label{continuous normalized NQS}
\ol \psi_{\mathrm{c}} = \frac{\psi_{\mathrm{c}}}{\|\psi_{\mathrm{c}}\|_{L^2}}  =\frac{\psi_{\mathrm{c}}}{\sqrt{\int_{\sum_{i=1}^n \|\vx_i\|_2^2 \le B}
|\psi_{\textup{c}}(\vx_1,\dots,\vx_n)|^2
\, d\vx_1 \cdots d\vx_n}},
\end{eqnarray}
and
\begin{eqnarray}
\label{discrete normalized NQS}
\ol\psi_{\textup{d}} =  \frac{\psi_{\textup{d}}}{\|\psi_{\textup{d}}\|_2} = \frac{\psi_{\textup{d}}}{\sqrt{\sum_{x_1,\dots,x_n}
|\psi_{\textup{d}}(x_1,\dots,x_n)|^2}}.
\end{eqnarray}
Before analyzing the generalization error of the predicted quantum wavefunctions/states, we first establish the following inequality. Since $\|\psi^\star\|=1$, the normalized estimator $\ol\psi=\psi/\|\psi\|$ is obtained by the radial projection onto the unit sphere. Consequently,
\begin{eqnarray}
\label{upper bound of normalized quantum wavefunction}
\|\ol\psi-\psi^\star\|
\le
\|\psi-\psi^\star\|.
\end{eqnarray}
Indeed, we can expand $\|\psi-\psi^\star\|^2 - \|\ol\psi-\psi^\star\|^2$ as follows:
\begin{eqnarray}
\label{expansion of two squared terms}
\|\psi-\psi^\star\|^2 - \|\ol\psi-\psi^\star\|^2 &\!\!\!\!=\!\!\!\!& \|\psi\|^2+1 - 2\operatorname{Re}\langle\psi,\psi^\star\rangle - \left( 2-\frac{2}{\|\psi\|} \operatorname{Re}\langle\psi,\psi^\star\rangle \right)\nonumber\\
&\!\!\!\!=\!\!\!\!& (\|\psi\|-1) \left( \|\psi\|+1- \frac{2\operatorname{Re}\langle\psi,\psi^\star\rangle} {\|\psi\|} \right).
\end{eqnarray}
Since $\operatorname{Re}\langle\psi,\psi^\star\rangle \le \|\psi\|$ by the Cauchy--Schwarz inequality, we have $\|\psi\|+1- \frac{2\operatorname{Re}\langle\psi,\psi^\star\rangle}{\|\psi\|} \ge \|\psi\|-1$.
Substituting this inequality into the above identity gives
\begin{eqnarray}
\label{expansion of two squared terms1}
\|\psi-\psi^\star\|^2 - \|\ol\psi-\psi^\star\|^2 \ge (\|\psi\|-1)^2 \ge 0,
\end{eqnarray}
which proves \eqref{upper bound of normalized quantum wavefunction}.Consequently, applying \eqref{upper bound of normalized quantum wavefunction} to the continuous and discrete settings yields
\begin{eqnarray}
\label{ineqaulity of continuous on the sphere}
\|\ol \psi_{\mathrm{c}} - \psi_{\mathrm{c}}^\star\|_{L^2}^2 &\!\!\!\!\leq\!\!\!\!& \|\psi_{\mathrm{c}} - \psi_{\mathrm{c}}^\star \|_{L^2}^2,\\
\label{ineqaulity of discrete on the sphere}
\|\ol \psi_{\mathrm{d}} - \psi_{\mathrm{d}}^\star\|_{2}^2 &\!\!\!\!\leq\!\!\!\!& \|\psi_{\mathrm{d}} - \psi_{\mathrm{d}}^\star \|_{2}^2.
\end{eqnarray}

We next combine \Cref{thm:generalization_transformer_nqs} with \eqref{ineqaulity of continuous on the sphere} to derive an upper bound on the expected error of the normalized continuous model. In particular, we obtain
\begin{eqnarray}
\label{upper bound of normalized quantum wavefunction expectation1}
\E[\|\ol \psi_{\mathrm{c}} - \psi_{\mathrm{c}}^\star\|_{L^2}^2] &\!\!\!\!\leq\!\!\!\!& \E[\|\psi_{\mathrm{c}} - \psi_{\mathrm{c}}^\star \|_{L^2}^2]\nonumber\\
&\!\!\!\!\leq\!\!\!\!& \E\bigg[ \int_{\sum_{i=1}^n \|\vx_i\|_2^2 \le B}
|\psi_{\textup{c}}(\vx_1,\dots,\vx_n) - \psi_{\mathrm{c}}^\star(\vx_1,\dots,\vx_n)|^2
\, d\vx_1 \cdots d\vx_n  \bigg]\nonumber\\
&\!\!\!\!=\!\!\!\!&  \int_{\sum_{i=1}^n \|\vx_i\|_2^2 \le B}
\E\big[|\psi_{\textup{c}}(\vx_1,\dots,\vx_n) - \psi_{\mathrm{c}}^\star(\vx_1,\dots,\vx_n)|^2\big]
\, d\vx_1 \cdots d\vx_n\nonumber\\
&\!\!\!\!=\!\!\!\!& O\bigg(\frac{\pi^{\frac{nD}{2}}}{\Gamma\!\left(\frac{nD}{2}+1\right)}
B^{\frac{nD}{2}}\bigg(  \sqrt{\frac{\log N}{N}} + \frac{nDB\max\{\log(BL), \log(B\sqrt{N}/\sqrt{\log N}) \}}{L} \bigg) \bigg)\nonumber\\
&\!\!\!\!=\!\!\!\!& O\bigg( \bigg(\frac{2\pi e B}{nD} \bigg)^{\frac{nD}{2}} \cdot \bigg(  \sqrt{\frac{\log N}{N}} + \frac{nDB\max\{\log(BL), \log(B\sqrt{N}/\sqrt{\log N}) \}}{L} \bigg) \bigg),
\end{eqnarray}
where the last line uses $\Gamma(z+1) \approx \sqrt{2\pi z}\left(\frac{z}{e}\right)^z$. A similar argument applies to the discrete case, where
\begin{eqnarray}
\label{upper bound of normalized quantum state expectation1}
\E[\|\ol\psi_{\textup{d}} - \psi_{\mathrm{d}}^\star\|_2^2] &\!\!\!\!\leq\!\!\!\!& \E[\|\psi_{\textup{d}} - \psi_{\mathrm{d}}^\star\|_2^2]\nonumber\\
&\!\!\!\!\leq\!\!\!\!& \E\bigg[ \sum_{x_1,\dots,x_n} |\psi_{\textup{d}}(x_1,\dots,x_n) - \psi_{\textup{d}}^\star(x_1,\dots,x_n)|^2\bigg]\nonumber\\
&\!\!\!\!=\!\!\!\!& O\bigg( d^n \cdot \bigg(\sqrt{\frac{\log N}{N}} + \frac{nd\max\{\log(dL), \log(d\sqrt{N}/\sqrt{\log N}) \}}{L}\bigg) \bigg).
\end{eqnarray}

Moreover, we relate the Frobenius norm of density operators to the $L^2$ distance of the underlying wavefunctions. Specifically,
\begin{eqnarray}
\label{upper bound of continuous F norm}
&\!\!\!\! \!\!\!\!&\|\ol \psi_{\mathrm{c}}\ol \psi_{\mathrm{c}}^\dagger - \psi_{\mathrm{c}}^\star{\psi_{\mathrm{c}}^\star}^\dagger \|_F\nonumber\\
&\!\!\!\! = \!\!\!\!&\|(\ol \psi_{\mathrm{c}} - \psi_{\mathrm{c}}^\star  )\ol \psi_{\mathrm{c}}^\dagger + \psi_{\mathrm{c}}^\star(\ol \psi_{\mathrm{c}} - {\psi_{\mathrm{c}}^\star})^\dagger \|_F\nonumber\\
&\!\!\!\! = \!\!\!\!&\|\ol \psi_{\mathrm{c}} - \psi_{\mathrm{c}}^\star\|_{L^2} \|\ol \psi_{\mathrm{c}}\|_{L^2} + \|\psi_{\mathrm{c}}^\star\|_{L^2}\|\ol \psi_{\mathrm{c}} - \psi_{\mathrm{c}}^\star\|_{L^2}\nonumber\\
&\!\!\!\! = \!\!\!\!&2\|\ol \psi_{\mathrm{c}} - \psi_{\mathrm{c}}^\star\|_{L^2}.
\end{eqnarray}
Similarly, in the discrete setting,
\begin{eqnarray}
\label{upper bound of discrete F norm}
&\!\!\!\! \!\!\!\!&\|\ol \psi_{\mathrm{d}}\ol \psi_{\mathrm{d}}^\dagger - \psi_{\mathrm{d}}^\star{\psi_{\mathrm{d}}^\star}^\dagger \|_F\leq 2\|\ol \psi_{\mathrm{d}} - \psi_{\mathrm{d}}^\star\|_{2}.
\end{eqnarray}
Combining \eqref{upper bound of normalized quantum wavefunction expectation1} and
\eqref{upper bound of normalized quantum state expectation1}, we obtain
\begin{eqnarray}
\label{upper bound of normalized quantum wavefunction expectation2}
\E[\|\ol \psi_{\mathrm{c}}\ol \psi_{\mathrm{c}}^\dagger - \psi_{\mathrm{c}}^\star{\psi_{\mathrm{c}}^\star}^\dagger \|_F^2] =  O\bigg( \bigg(\frac{B}{nD} \bigg)^{\frac{nD}{2}} \cdot \bigg(  \sqrt{\frac{\log N}{N}} + \frac{nDB\max\{\log(BL), \log(B\sqrt{N}/\sqrt{\log N}) \}}{L} \bigg) \bigg),
\end{eqnarray}
and
\begin{eqnarray}
\label{upper bound of normalized quantum state expectation2}
\E[\|\ol \psi_{\mathrm{d}}\ol \psi_{\mathrm{d}}^\dagger - \psi_{\mathrm{d}}^\star{\psi_{\mathrm{d}}^\star}^\dagger \|_F^2] = O\bigg( d^n \cdot \bigg(\sqrt{\frac{\log N}{N}} + \frac{nd\max\{\log(dL), \log(d\sqrt{N}/\sqrt{\log N}) \}}{L}\bigg) \bigg).
\end{eqnarray}
Finally, combining the above argument with the analysis of \eqref{Sufficient condition for output continuous} completes the proof.

\section{Auxiliary Materials}

According to \cite[Lemma 5]{barron1993universal} and \cite{li2025transformers}, we can directly obtain
\begin{lemma}
\label{lem:approx-by-phi}
Let $g : [-a,a] \to \R$ be any function satisfying $|g'(z)| \leq a$ and $|g(-1)| \leq 1$. Then, for any $\tau_k > 0$, there exists a function $\gamma(\cdot)$ such that
\begin{eqnarray}
\label{eq:phi-approx}
\big|
g(z) - \int \gamma(t)\, \phi\!\big(\tau_k (z+t)\big)\,\mathrm{d}t
\big|
\;\leq\;
3a\,\delta_{\tau_k},
\end{eqnarray}
where $\delta_{\tau_k}$ is defined in~\eqref{The definition of delta function}, $\phi(\cdot) $ is defined in \eqref{smoothed_step_function} and $\int |\gamma(t)|\,\mathrm{d}t \;\leq\; 3$.
\end{lemma}

Based on \Cref{lem:approx-by-phi} with \cite[Lemma 5]{li2025transformers}, we can further derive
\begin{lemma}
\label{approximation theorem in the continuous case}
For any $g_k \in \calG$, any $\vz$ obeying $\|\vz\|_2 \le \sqrt{B}$, any positive integer $m$, and any $\tau_k > 2$, we have
\begin{eqnarray}
    \label{The upper bound of approximation term g_k}
    \big| {g_k}(\vz) - {g_k}({\bf 0}) - g^{\mathrm{approx}}_k (\vz) \big|
    \le C_{g_k} \sqrt{B} \bigg( 3\delta_{\tau_k} + \frac{\pi}{m} \bigg),
\end{eqnarray}
where for any $g_k \in \calG$ we define
\begin{eqnarray}
    \label{The definition of approximate function g}
    g^{\mathrm{approx}}_k (\vz) := \sum_{i=1}^m \int_{\boldsymbol{\omega} \ne {\bf 0}} \int_{t} A_{g_k}(i,t,\boldsymbol{\omega})\,\phi\! \big( \tau_k(\|\boldsymbol{\omega}\|_2^{-1} \boldsymbol{\omega}^\top \vz - t ) \big) \Lambda(i, \mathrm{d}t, \mathrm{d}\boldsymbol{\omega}),
\end{eqnarray}
where $\phi(\cdot) $ is defined in \eqref{smoothed_step_function}. The measure $\Lambda(i, \mathrm{d}t, \mathrm{d}\boldsymbol{\omega})$ is a probability measure on $[m] \times \R \times \R^{nD}$ independent of $g_k$, while $A_{g_k}(i,t,\boldsymbol{\omega})$ denotes some weight function depending on $g_k,i,t,\boldsymbol{\omega}$ such that
\begin{eqnarray}
    \label{Upper bound of constant A}
    |A_{g_k}(i,t,\boldsymbol{\omega})|&\!\!\!\!\le\!\!\!\!& 3m\big( C_{g_k} - \sup_{\wt g_k \in \calG} |\wt g_k({\bf 0})| \big),\\
    \label{Upper bound of constant A expectation}
    E_{(i,t,\boldsymbol{\omega})\sim\Lambda}[A_{g_k}(i,t,\boldsymbol{\omega})]
    &\!\!\!\!\le\!\!\!\!&
    3\big( C_{g_k} - \sup_{\wt g_k \in \calG} |\wt g_k({\bf 0})| \big).
\end{eqnarray}
\end{lemma}


%

\end{document}